\documentclass[10pt,journal,final,twocolumn]{IEEEtran}

\usepackage{epsfig,graphicx,subfigure,psfrag,amsmath,cases,bm}
\usepackage{latexsym,amssymb,algorithm,mathtools}
\usepackage{algorithmic}
\usepackage{color}
\usepackage{url}
\usepackage{scrtime}
\usepackage{stfloats}
\usepackage{tablefootnote}
\usepackage{cite}
\usepackage{amsfonts,mathabx}
\usepackage{textcomp}
\usepackage{amsthm}
\usepackage{geometry}
\usepackage{enumitem}
\newcommand{\subparagraph}{}
\usepackage{titlesec}

\allowdisplaybreaks

\renewcommand{\thesubsubsection}{\arabic{subsubsection}}

\titleformat{\subsubsection}[runin]{\itshape}{\thesubsubsection)}{1em}{}
\titlespacing*{\subsubsection}{\parindent}{0pt}{*1}
\titlespacing*{\section}{0pt}{0.1\baselineskip}{0.2\baselineskip}
\titlespacing*{\subsection}{0pt}{0.1\baselineskip}{0.2\baselineskip}

\usepackage{xcolor,capt-of}
\usepackage{multirow}
\makeatletter
\newcommand*{\rom}[1]{\expandafter\@slowromancap\romannumeral #1@}
\makeatother

\DeclareMathOperator{\mino}{minimize}
\newtheorem{remark}{Remark}
\newtheorem{lemma}{Lemma}
\newtheorem{theorem}{Theorem}
\makeatletter
\def\thm@space@setup{\thm@preskip=0pt
\thm@postskip=1pt}
\makeatother

\IEEEaftertitletext{\vspace{-12mm}}

\title{Globally Optimal Movable Antenna-Enabled Multiuser Communication: Discrete Antenna Positioning, Power Consumption, and Imperfect CSI }
\author{ Yifei Wu ~\IEEEmembership{Graduate Student Member,~IEEE,}, Dongfang Xu~\IEEEmembership{Member,~IEEE,}, Derrick Wing Kwan Ng~\IEEEmembership{Fellow,~IEEE,}, Wolfgang Gerstacker~\IEEEmembership{Senior Member,~IEEE,}, and Robert Schober~\IEEEmembership{Fellow,~IEEE,}

\thanks{

Yifei Wu, Robert Schober, and Wolfgang Gerstacker are with the Institute for Digital Communications, Friedrich-Alexander-University Erlangen-N\"urnberg (FAU), Germany (email:\{yifei.wu, robert.schober, wolfgang.gerstacker\}@fau.de).

Dongfang Xu is with the Department of Electronic and Computer Engineering, The Hong Kong University of Science and Technology, Hong Kong (e-mail: eedxu@ust.hk).

D. W. K. Ng is with the School of Electrical Engineering and Telecommunications, the University of New South Wales, Australia (email: w.k.ng@unsw.edu.au).
}}

\begin{document}
\maketitle
\begin{abstract}
    Movable antennas (MAs) represent a promising paradigm to enhance the spatial degrees of freedom of conventional multi-antenna systems by dynamically adapting the positions of antenna elements within a designated transmit area. In particular, by employing electro-mechanical MA drivers such as stepper motors, the positions of the MA elements can be discretely adjusted to shape a favorable spatial correlation for improving system performance.
    Although preliminary research has explored beamforming designs for MA-enabled systems, the intricacies of the power consumption and the precise positioning of MA elements are not well understood, yet. Moreover, the assumption of perfect channel state information (CSI) adopted in the current literature is generally impractical due to the significant pilot overhead and the extensive time required for acquiring close-to-perfect CSI. To address these challenges, in this paper, we model the motion of MA elements through discrete steps and quantify the associated power consumption as a function of these movements. 
    Furthermore, by leveraging the properties of the MA channel model, we introduce a novel CSI error model tailored for MA-enabled systems that facilitates robust resource allocation design. In particular, we jointly optimize the beamforming and the MA positions at the base station (BS) for minimization of the total BS power consumption, encompassing both radiated power and MA motion power, while guaranteeing a minimum required signal-to-interference-plus-noise ratio for each user. To this end, novel algorithms exploiting the branch and bound (BnB) method are developed to obtain the globally optimal solution for perfect and imperfect CSI, respectively. Moreover, to support practical real-time implementation, we propose low-complexity suboptimal algorithms with guaranteed convergence by leveraging successive convex approximation (SCA). Our numerical results validate the global optimality of the proposed BnB-based algorithms for both CSI scenarios. Furthermore, we unveil that both proposed SCA-based algorithms approach the optimal performance of the BnB-based algorithms within only a few iterations, thus highlighting their practical advantages. Additionally, we show that compared to the state-of-the-art approach, the proposed low-complexity SCA-based schemes achieve considerable performance gains, especially in high-load systems with a small number of antenna elements. \vspace{-2mm}
\end{abstract}
{\color{black}
}
\section{Introduction}
Multiple-input multiple-output (MIMO) transmission has been widely recognized as a key technique to satisfy the skyrocketing data traffic demands of sixth-generation (6G) networks. Through the deployment of multiple antennas, MIMO can effectively utilize the spatial resources of wireless channels to deliver significant enhancements in performance, including higher data transmission rates\cite{mietzner2009multiple}, improved physical layer security\cite{tsai2014power}, and the realization of innovative paradigms such as integrated sensing and communication \cite{xu2022robust}. 
However, despite these benefits, conventional MIMO systems suffer from elevated hardware costs and computational complexity due to the requirement of equipping multiple parallel power-hungry radio frequency (RF) chains\cite{mietzner2009multiple}.
To address this challenge, antenna selection (AS) has been advocated as a realistic strategy for implementing practical MIMO systems. The crux of AS is to strategically choose a small subset of antennas with favorable channel characteristics from a larger pool of candidate antennas. By delicately choosing favorable antennas, AS maintains the diversity gain of MIMO systems while reducing the number of required RF chains. However, traditional MIMO systems, even those employing AS, encounter limitations due to the fixed positions of their antennas. This fixed arrangement restricts their ability to fully exploit spatial variations in the channel across the transmit area, potentially undermining system performance.

To fully capitalize on the spatial variations of wireless channels across a designated spatial transmit area, holographic MIMO (HMIMO) has emerged as a recent groundbreaking development \cite{huang2020holographic}. Specifically, HMIMO surfaces are composed of numerous miniature passive elements that are spaced at sub-wavelength intervals. These elements can be electronically controlled to manipulate the electromagnetic properties of transmitted or reflected waves, offering a novel approach for channel optimization and signal enhancement \cite{huang2020holographic}. Therefore, the utilization of zero-spacing continuous antenna elements in HMIMO allows for the exploitation of the full spatial degrees of freedom (DoFs) within a spatially continuous transmit area. Indeed, this unique configuration enables a more intricate control and manipulation of the electromagnetic field, enhancing the overall performance of the communication system. However, despite the high potential of HMIMO to revolutionize communications systems, significant challenges arise in its implementation due to the large required number of antenna elements. This challenge is further manifested in both channel estimation and data processing, limiting the practical viability of HMIMO\cite{huang2020holographic}. 

Motivated by the large number of spatial DoFs potentially offered by HMIMO surfaces, the novel MIMO concept of movable antennas (MAs) has emerged as a promising technology, bridging the gap between HMIMO and traditional MIMO \cite{zhu2022modeling}. In MA-enabled systems, each antenna element is connected to an RF chain through a flexible cable, allowing its physical position to be dynamically adapted within a specified spatial region by an electro-mechanical driver \cite{ma2022mimo} or using liquid metals \cite{wong2020fluid}. Thus, this unique capability endows MA-enabled systems with the flexibility to adjust antenna positions in real-time to establish favorable spatial antenna correlations thereby improving the performance of MIMO systems \cite{zhu2023movable,zhu2023movable_chanllenges}. Moreover, since MA systems necessitate only a small number of antenna elements and RF chains to exploit the available DoFs, the computational complexity required for signal processing is significantly reduced compared to HMIMO systems \cite{zhu2022modeling,ma2022mimo}. 

{\color{black}Similar to MA, hybrid beamforming (HBF) is also capable of enhancing the signal directionality and the spatial DoFs in wireless communication systems. 
 In particular, HBF utilizes a combination of analog beamforming in the RF domain and digital beamforming in the baseband, where the analog beamformer aims to establish a favorable channel correlation to boost MIMO system performance. Although both the MA and HBF concepts aim to exploit specific hardware for beam shaping using a limited number of RF chains, they employ different implementation strategies. MAs utilize electromechanical drivers to adjust the position of antenna elements physically, while HBF exclusively uses electronic phase shifters for beam steering. Furthermore, HBF generally necessitates a large number of antenna elements, high-resolution phase shifters, and energy-hungry signal processing components. In contrast, MAs circumvent the need for additional RF phase shifters and an excessive number of antenna elements, which could result in reduced complexity and hardware cost. }

Based on their hardware architecture, state-of-the-art MA systems can be categorized into two types: mechanical MA systems \cite{ma2022mimo,zhu2023movable,wu2023movable} and fluid antenna systems (FAS) \cite{wong2020fluid,wong2021fluid}. In particular, mechanical MA systems employ electro-mechanical devices such as stepper motors to move the antennas. In contrast, FAS feature antennas that can switch positions almost instantaneously in a small linear space, facilitated by the use of liquid metals. Various initial works demonstrated considerable potential for both types of MA systems by jointly optimizing MA positions and base station (BS) beamforming. For instance, the concept of fluid antenna multiple access (FAMA) was introduced in \cite{wong2021fluid} using a single fluid antenna (FA) for each mobile user. In FAMA, the positions of the FAs at different users are strategically selected to achieve favorable channel conditions with reduced interference, thereby resulting in massive capacity gains. However, due to the inherent physical properties of liquid materials, FAS can only support a single FA in motion along a one-dimensional line. Therefore, FAS cannot fully exploit variations in the wireless channel across the two-dimensional (2D) or three-dimensional (3D) spatial domains, thereby restricting their effectiveness in complex environmental conditions.
To tackle this limitation, several works have considered electro-mechanical systems as MA drivers to enable flexible movement in 2D or 3D transmit regions \cite{wu2023movable, ma2022mimo}. For instance, in \cite{ma2022mimo}, the authors investigated scenarios in which both the BS and multiple users were equipped with MAs. To jointly optimize the MA positions and beamforming, an alternating optimization (AO)-based suboptimal algorithm was introduced. Similarly, for the multiuser mechanical MA-enabled uplink communication system in \cite{zhu2023movable}, the authors proposed suboptimal joint designs based on zero-forcing (ZF) and minimum mean square error (MMSE) combining, respectively. In particular, the BS employs ZF or MMSE combining, while the MA positions are adjusted with a gradient descent (GD) method. However, both \cite{ma2022mimo} and \cite{zhu2023movable} relied on the optimistic assumption that the electro-mechanical MA driver can position the MAs freely within a given region. In contrast, prototype designs of MA-enabled systems, reported in, e.g., \cite{zhuravlev2015experimental} and \cite{basbug2017design}, reveal that the motion control of the employed electro-mechanical devices is limited to discrete adjustments with finite precision. Thus, the transmit area is spatially quantized \cite{basbug2017design}, providing a finite spatial resolution rather than the infinite resolution previously assumed in \cite{zhu2022modeling,ma2022mimo,zhu2023movable}.
{\color{black} Therefore, several works have considered practical discrete MA movement \cite{mei2024movable,wei2024joint,shao20256d}. However, these works do not account for the power consumption introduced by electro-mechanical MA drivers. In fact, to facilitate rapid position adjustment of MA elements, practical MA-enabled systems necessitate the application of high-speed electro-mechanical devices as MA drivers, leading to significant power consumption. Consequently, this movement-dependent power consumption, which depends on the MA position displacement, must be considered for resource allocation design.} 
In addition, existing MA-enabled system designs are based on suboptimal frameworks, e.g.,  the AO-based algorithm in \cite{ma2022mimo} and the GD-based method in \cite{zhu2023movable}. These suboptimal approaches generally compromise global optimality, and their performance heavily depends on the chosen initial point, which may lead to unsatisfactory system performance. This highlights the need for optimization methods that can ensure global optimality and convergence. 

Recently, in the conference version \cite{wu2023movable} of this paper, a novel algorithm based on the generalized Benders decomposition framework was developed for attaining the globally jointly optimum BS beamformer and MA positions that minimize the radiated power of multiuser MA-enabled systems. However, the promised performance in \cite{wu2023movable} relies on the acquisition of perfect CSI for all MA candidate positions. In fact, since the transmit area of MA-enabled systems is quantized to numerous discrete positions, it is challenging to accurately estimate the channels between all MA candidate positions and all users by conventional channel estimation methods \cite{xiao2024channel}. Several novel channel estimation schemes based on compressed sensing (CS) are proposed for MA-enabled systems in \cite{xiao2024channel,ma2023compressed}. However, the state-of-the-art estimation methods in \cite{xiao2024channel,ma2023compressed} require CSI feedback from users to BS, which introduces inevitable quantization error to instantaneous CSI. 
{\color{black} To avoid instantaneous CSI estimation, the authors in \cite{chen2023joint,ye2023fluid,hu2024two,hu2024movable} have investigated statistical beamforming for MA-enabled systems assuming ideal continuous MA movement. In particular, robust joint beamforming and MA movement based on the statistical CSI of the MA channels was studied in \cite{chen2023joint,ye2023fluid}. Moreover, the authors in \cite{hu2024two} advocated a novel two-time-scale framework that optimizes the antenna positions over a long time frame while designing the transmit and receive beamformers based on perfect instantaneous CSI for a short time frame. Additionally, the authors in \cite{hu2024movable} proposed a secure MA beamforming and positioning strategy based on the statistical CSI of eavesdroppers and the perfect instantaneous CSI of the users.}
However, for practical MA systems with imperfect instantaneous CSI and discrete MA movement, neither globally optimal nor efficient suboptimal robust resource allocation frameworks are available in the literature. Motivated by the above discussion, in this paper, we investigate for the first time the globally optimal joint design of BS beamforming and MA positioning for multiuser MA-enabled systems with discrete MA movement for both perfect and imperfect CSI. The main contributions of this paper can be summarized as follows:\\[-15pt]
\begin{itemize}[leftmargin=10pt]
\item In this work, we consider an MA-enabled multiuser MISO downlink system, where the movement of the antennas is enabled by electro-mechanical devices. In particular, for the first time, we model MA movement as a discrete motion and propose a corresponding power consumption model. Next, we study the resource allocation problem for the considered MA-enabled system. The objective of the proposed resource allocation design is minimizing the total power consumption at the BS, which accounts for the aggregated power comprising the radiated power and the power required for MA motion, while guaranteeing a minimum signal-to-interference-plus-noise ratio (SINR) for each user. 
\item We first investigate the resource allocation design for the ideal case of perfect CSI. The joint BS beamforming and MA position design is formulated as a mixed integer nonlinear programming (MINLP) problem. To tackle this challenging problem, we propose a series of mathematical transformations aimed at facilitating the design of an iterative algorithm based on the branch and bound (BnB) method, which is guaranteed to converge to the globally optimal solution. This yields a performance upper bound for resource allocation design in MA systems with perfect CSI and serves as a critical benchmark for evaluating the performance of any corresponding suboptimal algorithm, such as \cite{zhu2023movable, ma2022mimo}.  Moreover, to strike a balance between performance and complexity, we develop a novel and efficient suboptimal scheme by leveraging successive convex approximation (SCA).
\item We introduce a novel CSI error model based on the field response channel model in \cite{zhu2022modeling}. Then, building on the framework developed for perfect CSI, we extend our design to the more practical case of imperfect CSI. Specifically, we transform the resulting robust SINR constraints into a set of linear matrix inequality (LMI) constraints. Furthermore, we introduce a set of auxiliary variables to formulate an MINLP problem for robust resource allocation design. Then, we extend the BnB-based and SCA-based methods, originally developed for the case of perfect CSI, to obtain the global optimum and a local optimum for imperfect CSI, respectively. 
\item By analyzing the convergence behavior of the proposed BnB-based algorithms, we confirm their capability to find global optimal solutions for the considered problems for both perfect and imperfect CSI. Furthermore, both proposed SCA-based algorithms achieve local optimality, closely approaching the globally optimal upper bound established by the BnB-based algorithms. Notably, the SCA-based algorithms exhibit a faster convergence, although at the cost of a slightly higher total power consumption compared to the optimal BnB-based algorithms, offering a practical balance between system performance and complexity.\\[-15pt]
\end{itemize}
The remainder of this paper is organized as follows: In Section \rom{2}, we introduce the system model for the considered MA-enabled multiuser multiple-input single-output (MISO) system with a spatially discrete transmit area and formulate the corresponding resource allocation problem. In Sections \rom{3} and \rom{4}, the globally optimal and suboptimal solutions for the MA positions and the BS beamformer are developed for perfect and imperfect CSI, respectively. Section \rom{5} evaluates the performance of the proposed MA-enabled system designs via numerical simulations, and Section \rom{6} concludes this paper.

\textit{Notation:} 
Vectors and matrices are denoted by boldface lower case and boldface capital letters, respectively. $\mathbb{R}^{N\times M}$ and $\mathbb{C}^{N\times M}$ represent the spaces of $N\times M$ real-valued and complex-valued matrices, respectively. $\mathrm{Re}\{\cdot\}$ and $\mathrm{Im}\{\cdot\}$ denote the real and imaginary parts of a complex number, respectively. $|\cdot|$, $||\cdot||_2$ and $||\cdot||_F$ stand for the absolute value of a complex scalar, the $l_2$-norm of a vector and the Frobenius norm of a matrix, respectively. $(\cdot)^T$, $(\cdot)^*$, and $(\cdot)^H$ denote the transpose, conjugate, and conjugate transpose of their arguments, respectively. $\mathbf{I}_{N}$ denotes the identity matrix of dimension $N$. $\mathrm{Tr}(\cdot)$ and $\lambda_{\mathrm{max}}(\cdot)$ are the trace and the largest eigenvalue of the input argument, respectively. $\mathbf{0}_{L}$ and $\mathbf{1}_L$ represent the all-zeros and all-ones vector of length $L$, respectively. $\mathbf{A}\hspace{-0.5mm}\succeq\hspace{-0.5mm}\mathbf{0}$ indicates that $\mathbf{A}$ is a positive semidefinite matrix. $\mathrm{diag}(\mathbf{a})$ denotes a diagonal matrix whose main diagonal elements are given by the entries of vector $\mathbf{a}$. $\mathbb{E}[\cdot]$ refers to statistical expectation. ${N}\choose{K}$ is the binomial coefficient of monomial $x^K$ in the expansion of $(1 + x)^N$. 
\section{System Model}
\subsection{System Model and Channel Model}
\begin{figure}\vspace{-2mm}
    \centering
    \includegraphics[width=1.8in]{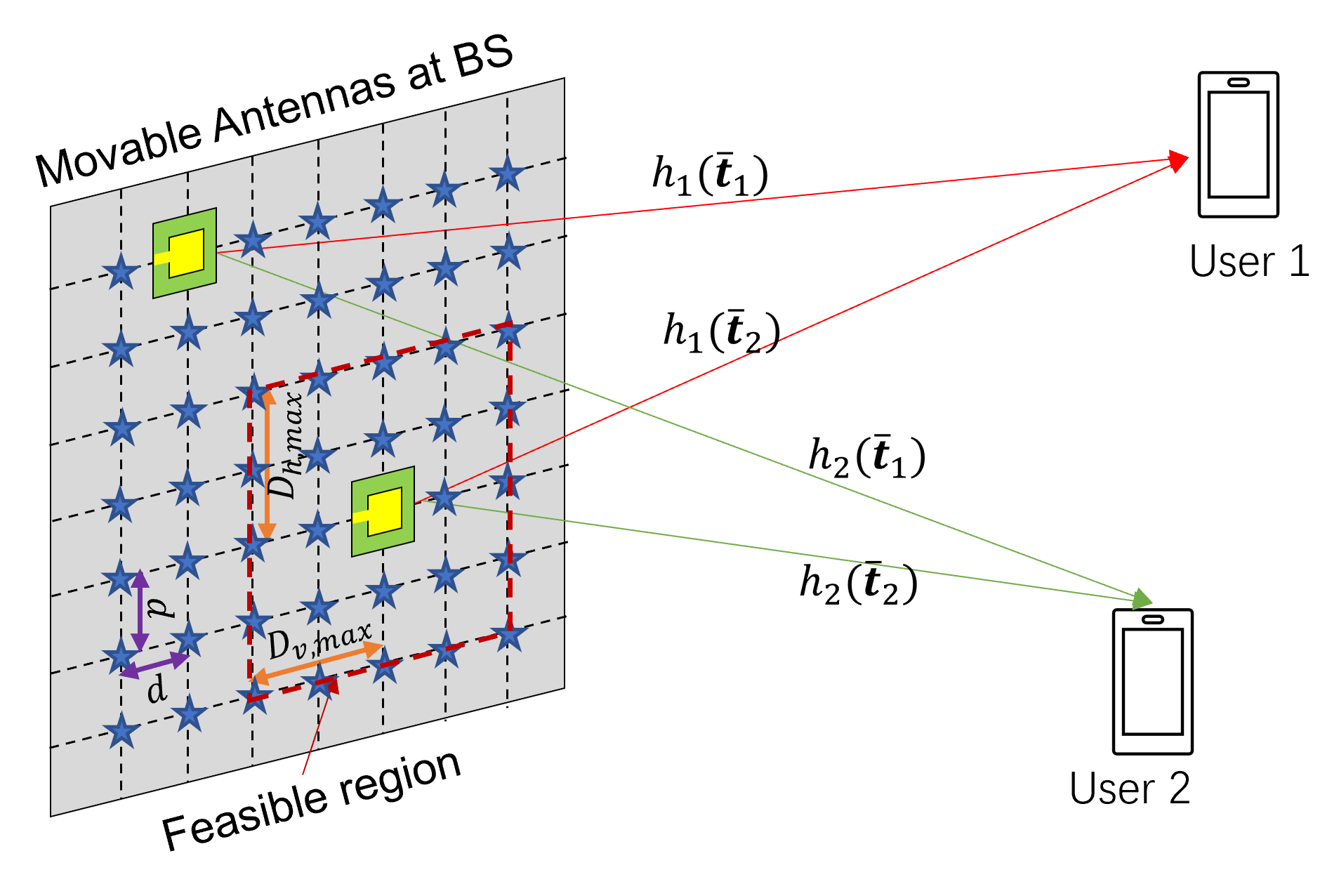}
    \caption{Transmission from $M=2$ MA elements with $N=49$ possible discrete positions to $K=2$ users (markers ``$\bigstar$" indicate discrete antenna positions).}\vspace{2mm}
    \label{fig:MA_system_model}
\end{figure}
We consider a multiuser wireless communication system where a BS, equipped with $M$ MA elements, serves $K$ single-antenna users in the downlink. 
The positions of the MA elements can be adjusted simultaneously within a given 2D transmit area. 
Since practical electro-mechanical devices can only provide horizontal and vertical movement by a fixed constant increment $d$ in each step \cite{zhuravlev2015experimental}, \cite{basbug2017design}, the transmit area of the MA-enabled communication system is quantized \cite{basbug2017design}\footnote{The value of step size $d$ depends on the precision of the employed actuators of the electro-mechanical devices and may vary in different systems \cite{zhuravlev2015experimental,basbug2017design}. \vspace{-1mm}}. We collect the $N$ possible discrete positions of the MAs in set $\mathcal{P}=\{\mathbf{p}_1,\cdots, \mathbf{p}_N\}$, where the distance between neighboring positions is equal to $d$ in horizontal or vertical directions, as shown in Fig. \ref{fig:MA_system_model}. Here, $\mathbf{p}_n=[x_n,y_n]^T$ represents the $n$-th candidate position with horizontal coordinate $x_n$ and vertical coordinate $y_n$. 
Then, the feasible set of the position of MA element $m$, $\bar{\mathbf{t}}_m$, is given by $\mathcal{P}$, i.e., $\bar{\mathbf{t}}_m=[\bar{x}_m,\bar{y}_m]^T\in\mathcal{P}$. 

\begin{figure}\vspace{-0.3cm}
    \centering
    \includegraphics[width=2.2in]{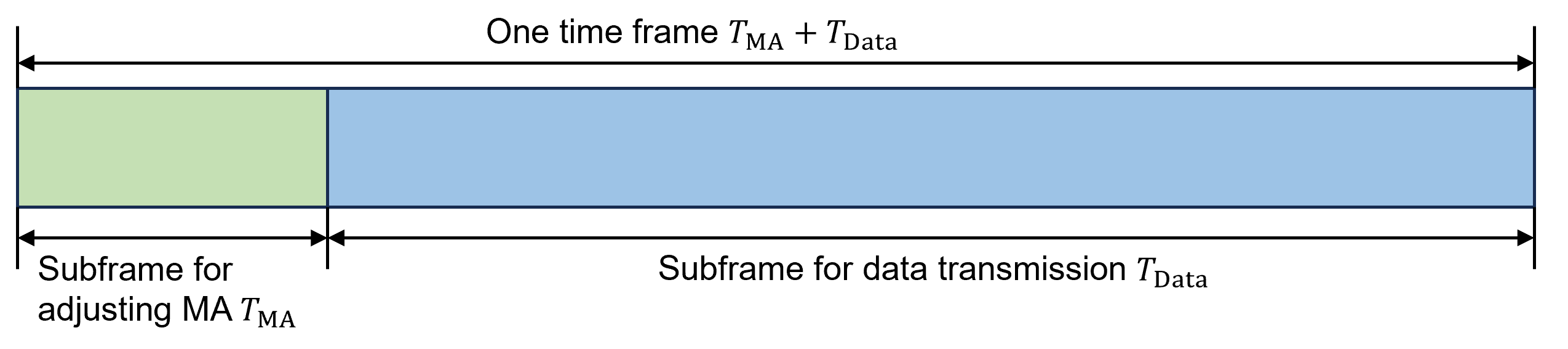}
    \caption{Time frame structure for the considered quasi-static fading scenario.}\vspace{2mm}
    \label{fig:frame}
\end{figure}
In this work, we consider a slow-fading channel and focus on one quasi-static fading block \cite{zhu2022modeling,zhu2023movable}. The slow-varying channel allows moving the MA elements to locations with favorable channel properties at the beginning of each coherence time frame. Hence, each coherence time frame is divided into two subframes, where in the first subframe, the MA positions are adjusted while the second subframe is used for data transmission. Here, $T_{\mathrm{MA}}$ and $T_{\mathrm{Data}}$ denote the time durations of the first and second subframes, respectively. Moreover, the physical channel conditions can be reconfigured by adapting the positions of the MA elements.
Then, we define $\widebar{\mathbf{T}}=[\bar{\mathbf{t}}_1,\cdots, \bar{\mathbf{t}}_M]$ as the collection of the positions of all MA elements. Then, the MIMO channel between the BS and the $K$ users $\mathbf{H}(\widebar{\mathbf{T}})=[\mathbf{h}_1(\widebar{\mathbf{T}}),\cdots, \mathbf{h}_K(\widebar{\mathbf{T}})]^H$ is a function of $\widebar{\mathbf{T}}$, where $\mathbf{h}_k^H(\widebar{\mathbf{T}})=[h_k(\bar{\mathbf{t}}_1),\cdots,h_k(\bar{\mathbf{t}}_M)]$ denotes the channel vector between the BS and user $k$. Here, $h_k(\bar{\mathbf{t}}_m)\in\mathbb{C}$ is the channel coefficient 
between MA element $m$ at position $\bar{\mathbf{t}}_m$ and user $k$. Thus, the received signal $y_k$ at user $k$ is given by
\begin{equation}
y_k=\mathbf{h}_k^H(\widebar{\mathbf{T}})\mathbf{w}_ks_k+\sum_{k'\neq k}\mathbf{h}_k^H(\widebar{\mathbf{T}})\mathbf{w}_{k'}s_{k'}+n_k,
\end{equation}
where $s_k\in\mathbb{C}$ denotes the symbol transmitted to user $k$, $\mathbb{E}[|s_k|^2]=1$, $\mathbb{E}[s_k^*s_{k'}]=0$, $k\neq k'$, $\forall k,k'\in\{1,\cdots,K\}$. Here, we define vector $\mathbf{s}=[s_1,\cdots,s_K]^T\in\mathbb{C}^{K\times 1}$ to collect the symbols transmitted to all users. Vector $\mathbf{w}_k\in \mathbb{C}^{M\times 1}$ is the beamforming vector of user $k$, and $n_k\in\mathbb{C}$ stands for the additive white Gaussian noise at user $k$ with zero mean and variance $\sigma_k^2$. To simplify the notation, we define sets $\mathcal{K}=\{1,\cdots, K\}$ and $\mathcal{M}=\{1,\cdots, M\}$ to collect the indices of the users and the MA elements, respectively.

In this work, we adopt the field response channel model proposed in \cite{zhu2022modeling} for the considered system. In particular, the channel vector $\mathbf{h}_k(\widebar{\mathbf{T}})$ between the BS and user $k$ is contingent upon both the propagation environment and the positions of the $M$ MA elements. Since the size of the transmit area for the MA elements is much smaller than the communication distance between the BS and the users, the far-field condition generally holds for the BS-user channels\cite{zhu2022modeling,zhu2023movable,ma2022mimo}. Thus, the field response from the transmit area to the users can be modeled as a plane wave. In particular, adjusting the positions of the MA elements does not impact the angle of departure (AoD), the angle of arrival (AoA), and the amplitude of the multi-path components (MPCs) between the BS and each user, while the phase of the MPCs depends on MA position

Without loss of generality, the channel between the BS and user $k$ is assumed to comprise $L_k$ MPCs. Let $\theta_{k,l_k}$ and $\phi_{k,l_k}$ denote the elevation and azimuth AoDs of the $l_k$-th channel path of user $k$, respectively. The transmit field-response vector (FRV) for the channel between user $k$ and an MA element at position $\mathbf{p}_n$ is given as $\mathbf{g}_k(\mathbf{p}_n)=[e^{j\rho_{k,1}(\mathbf{p}_n)},\cdots,e^{j\rho_{k,L_k}(\mathbf{p}_n)}]^T$\cite{zhu2022modeling},
where $\rho_{k,l_k}(\mathbf{p}_n)=\frac{2\pi}{\lambda}\left((x_n-x_1)\cos\theta_{k,l_k}\sin\phi_{k,l_k}+(y_n-y_1)\sin\theta_{k,l_k}\right)$ denotes the phase difference of the $l_k$-th channel path between $\mathbf{p}_n$ and the reference position, i.e., $\mathbf{p}_1$. Here, $\lambda$ denotes the carrier wavelength. Since the users are equipped with a single antenna, the receive FRV for user $k$ is given by $\mathbf{1}_{L_k}^T$.
Thus, the channel coefficient $h_{k}(\mathbf{p}_n)\in\mathbb{C}$ between an MA element at position $\mathbf{p}_n$ and user $k$ is modeled as follows
\begin{equation}
    h_{k}(\mathbf{p}_n)=\mathbf{1}_{L_k}^T\bm{\Psi}_k\mathbf{g}_k(\mathbf{p}_n)=\bm{\psi}_k^H\mathbf{g}_k(\mathbf{p}_n),
\end{equation}
where matrix $\bm{\Psi}_k=\mathrm{diag}\{[\psi_{1,k},\cdots,\psi_{L_k,k}]^T\},\psi_{l_k,k}\in\mathbb{C},\forall l_k\in\{1,\cdots, L_k\}$, contains the path coefficients of all $L_k$ channel paths from the transmit area to user $k$ and $\bm{\psi}_k^H=\mathbf{1}_{L_k}^T\bm{\Psi}_k$ denotes the corresponding path coefficient vector (PCV). Moreover, we define an effective channel vector, $\hat{\mathbf{h}}_k\in\mathbb{C}^{N}$, that collects the channel coefficients between an MA element and user $k$ for all possible $\mathbf{p}_n\in\mathcal{P}$, i.e.,
\begin{equation}\notag
\begin{aligned}
    \hat{\mathbf{h}}_k&=[h_{k}(\mathbf{p}_1), \cdots, h_{k}(\mathbf{p}_N)]^H=\mathbf{G}_k^H\bm{\psi}_k,
\end{aligned}
\end{equation}
where $\mathbf{G}_k=[\mathbf{g}_k(\mathbf{p}_1),\cdots, \mathbf{g}_k(\mathbf{p}_N)]$ denotes the field response matrix (FRM) of user $k$. 
\subsection{MA Model}
In this paper, we take into account the physical antenna size of the MA elements for MA location design. In particular, two MA elements cannot be placed arbitrarily close to each other to avoid collisions during movement. Thus, the center-to-center distance between any pair of MA elements must exceed a specified minimum distance $D_{\mathrm{min}}$ \cite{ma2022mimo}, i.e.,
\begin{equation}\label{Min_dis_con}
    \|\bar{\mathbf{t}}_m-\bar{\mathbf{t}}_{m'}\|_2\geq D_{\mathrm{min}}, m\neq m', \forall m, m'\in\mathcal{M},
\end{equation}
where $\bar{\mathbf{t}}_m\hspace{-0.5mm}=\hspace{-0.5mm}[\bar{x}_m,\bar{y}_m]^T$ is the position of MA element $m$.

Furthermore, due to the limitation of electro-mechanical drivers, the moving speed of MA elements is constrained. In this paper, we assume that the adopted MA drivers provide line motion in both horizontal and vertical directions, as shown in \cite{basbug2017design,zhu2023movable_chanllenges}. In particular, two motion drivers are deployed for one MA element to enable simultaneous horizontal and vertical motion. For the sake of simplicity, we assume that the horizontal and vertical speeds of the MAs $v_{\mathrm{h,MA}}$ and $v_{\mathrm{v,MA}}$ are constant. Thus, for one transmission block, the maximum horizontal and vertical movement distances of each MA element during the first subframe are given by $D_{\mathrm{h,max}}=v_{\mathrm{h,MA}}T_{\mathrm{MA}}$ and $D_{\mathrm{v,max}}=v_{\mathrm{v,MA}}T_{\mathrm{MA}}$, respectively. Thus, the position of MA element $m$ must satisfy the following constraint:
\begin{equation}\label{Max_dis_cons}
    |\bar{x}_m-x_m^0|\leq D_{\mathrm{h,max}},\quad |\bar{y}_m-y_m^0|\leq D_{\mathrm{v,max}},
\end{equation}
where $[x_m^0, y_m^0]$ denote the position of MA element $m$ at the beginning of the first subframe.  

In this work, we assume that the power consumptions of the horizontal and vertical MA drivers are constant during MA motion, respectively.
{Thereby, the energy consumption of the $m$-th MA driver is proportional to the time interval needed for the horizontal and vertical movement of MA element $m$, i.e.,  $t_{\mathrm{m,h,MA}}=\frac{|\bar{x}_m-x_m^0|}{v_{\mathrm{h,MA}}}$ and $t_{\mathrm{m,v,MA}}=\frac{|\bar{y}_m-y_m^0|}{v_{\mathrm{v,MA}}}$. The overall energy consumption required for MA motion is given by $ E_{m,\mathrm{MA}}=P_{\mathrm{h,MA}}\frac{|\bar{x}_m-x_m^0|}{v_{\mathrm{h,MA}}}+P_{\mathrm{v,MA}}\frac{|\bar{y}_m-y_m^0|}{v_{\mathrm{v,MA}}},$
where $P_{\mathrm{h,MA}}$ and $P_{\mathrm{v,MA}}$ denote the power consumption of the horizontal and vertical MA drivers, respectively. 
}
\subsection{CSI Uncertainty}
{\color{black}
For optimization of the performance of the considered MA-enabled system, acquiring the CSI of the effective channel vectors $\hat{\mathbf{h}}_{k}$, $\forall k$, at the BS is essential. To efficiently acquire the effective channel vectors $\hat{\mathbf{h}}_k$ for all users, the state-of-the-art approaches in \cite{xiao2024channel,ma2023compressed} exploit compressed sensing (CS) techniques to estimate the key parameters of the MA channels and reconstruct the effective channel vectors based on these estimated parameters. In particular, for the channel estimation schemes proposed in \cite{xiao2024channel,ma2023compressed}, the users send several pilots, and the MA elements at the BS are moved over a number of predefined MA positions to receive the pilots. Then, the BS collects the received signals and estimates the azimuth and elevation AoDs of the channel paths, which are subsequently used to reconstruct FRM $\mathbf{G}_k$, $\forall k$. After the AoD estimation, the BS sends several pilots to the users for a number of predefined MA positions to estimate the PCV. Given the attained FRMs $\mathbf{G}_k$, the PCV estimation problem is reduced to a typical least-squares problem, which can be readily solved by the users. Finally, the estimated PCVs at the users are fed back to the BS to reconstruct the effective channel vectors. More details regarding channel estimation in MA-enabled systems are provided in \cite{ma2023compressed,xiao2024channel}.

In this paper, we consider a slow-fading channel and focus on the transmission of a single block over the quasi-static fading channel \cite{zhu2022modeling,zhu2023movable}, for which the instantaneous CSI has been estimated. 
Moreover, thanks to the quasi-static assumption, the FRMs of all users can be regarded as constant and can be accurately estimated at the BS \cite{xiao2024channel}. On the user side, the PCV is recovered via least-squares-based methods and subsequently fed back to the BS \cite{xiao2024channel}\footnote{\color{black}In this work, the FRM and PCV are assumed to be perfectly estimated at the BS and users, respectively. If the estimated FRM and PCV contain estimation errors, these estimation errors can also be characterized by the adopted norm-bounded model according to \cite[Lemma 1]{wiesel2007optimization}.}. However, due to the limited CSI feedback rate, the PCV feedback is subject to quantization errors \cite{zheng2009robust,shenouda2007convex}. 
In this work, we adopt a norm-bounded imperfect CSI error model to characterize these quantization errors.}
Specifically, the effective channel vector $\hat{\mathbf{h}}_{k}$ between the MA region and user $k$ is modeled as follows
\begin{equation}\label{CSI_model}
\begin{aligned}
\hspace{-2.5mm}\hat{\mathbf{h}}_{k}\hspace{-0.5mm}=\hspace{-0.5mm}\mathbf{G}_k^H({\bar{\bm{\psi}}}_{k}\hspace{-0.5mm}+\hspace{-0.5mm}\Delta {\bm{\psi}}_{k}),\,\Omega_{\Delta {\bm{\psi}}_{k}}\hspace{-1mm}=\hspace{-1mm}\{\Delta {\bm{\psi}}_{k}\hspace{-0.5mm}\in\hspace{-0.5mm}\mathbb{C}^{L_k}\hspace{-0.5mm}:\hspace{-0.5mm}\|\Delta {\bm{\psi}}_{k}\|_2^2\hspace{-0.5mm}\leq\hspace{-0.5mm} \epsilon_k^2\},
\end{aligned}
\end{equation}
where ${\bar{\bm{\psi}}}_{k}$ and $\Delta {\bm{\psi}}_{k}$ denote the estimate and the estimation error of the PCV of user $k$, respectively. The norm of the CSI error, $\Delta {\bm{\psi}}_{k}$, is bounded by constant $\epsilon_k$ and set $\Omega_{\Delta {\bm{\psi}}_{k}}$ contains all possible CSI errors satisfying the bounded norm condition.
\section{Algorithm Design for Perfect CSI}
In this section, assuming perfect CSI is available at the BS, we formulate the resource allocation problem for the considered MA-enabled systems and transform it into a more tractable MINLP problem. Then, we develop two algorithms based on the BnB and SCA methods to attain the global optimum and a local optimum of the reformulated MINLP problem, respectively.
\subsection{Problem Formulation}
We first rewrite $\mathbf{h}_{k}^H(\bar{\mathbf{T}})$ as $\mathbf{h}_k^H(\widebar{\mathbf{T}})=\hat{\mathbf{h}}_{k}^H\mathbf{B},$
where $\mathbf{B}=[\mathbf{b}_1,\cdots, \mathbf{b}_M]$ and $\mathbf{b}_m=\big[b_m[1],\cdots,b_m[N]\big]^T$ denotes the binary selection vector of MA element $m$. Here, $b_m[n]\in\left\{0,\hspace*{1mm}1\right\}$ with $\sum_{n=1}^{N}b_m[n]=1$ is a binary variable defining the position of MA element $m$. Note that $b_m[n]=1$ if and only if the $n$-th MA position in $\mathcal{P}$ is selected, i.e., $\bar{\mathbf{t}}_m=\mathbf{p}_n$. Thus, the received signal of user $k$ is given by
\begin{equation}
    y_k=\hat{\mathbf{h}}_{k}^H\mathbf{B}\sum_{k\in\mathcal{K}}\mathbf{w}_k{s}_k+n_k=\hat{\mathbf{h}}_{k}^H\mathbf{B}\mathbf{W}\mathbf{s}+n_k,\vspace{-1mm}
\end{equation}
where $\mathbf{W}=[\mathbf{w}_1,\cdots,\mathbf{w}_K]$ denotes the collection of the beamforming vectors of all users.
Therefore, the received SINR of user $k$ is given by\vspace*{-1mm}
\begin{equation}\vspace*{-1mm}
\mathrm{SINR}_k=\frac{|\hat{\mathbf{h}}_{k}^H\mathbf{B}\mathbf{w}_k|^2}{\sum_{k'\in\mathcal{K}\setminus{k}}|\hat{\mathbf{h}}_{k}^H\mathbf{B}\mathbf{w}_{k'}|^2+\sigma_k^2}.
\end{equation}
Next, we define distance matrix $\mathbf{D}\in\mathbb{C}^{N\times N}$, where element $D_{n,n'}$ in the $n$-th row and $n'$-th column of $\mathbf{D}$ denotes the distance between the $n$-th and the $n'$-th candidate positions in $\mathcal{P}$. Thus, the minimum distance constraint between any pair of MA elements in \eqref{Min_dis_con} can be reformulated as 
\begin{equation}
    \mathbf{b}_m^T\mathbf{D}\mathbf{b}_{m'}\geq D_{\mathrm{min}},\ m\neq m',\ \forall m, m'\in\mathcal{M}.
\end{equation}
Moreover, let vector $\mathbf{d}_{\mathrm{h,m}}=[|x_{1}-x_m^0|,\cdots,|x_{N}-x_m^0|]^T$ and $\mathbf{d}_{\mathrm{v,m}}=[|y_{1}-y_m^0|,\cdots,|y_{N}-y_m^0|]^T$ collect the horizontal and vertical distances between all candidate positions in $\mathcal{P}$ and the initial position of MA element $m$. Then, the maximal moving distance constraints in \eqref{Max_dis_cons} can be equivalently rewritten as 
\begin{equation}
    \mathbf{b}_m^T\mathbf{d}_{\mathrm{h,m}}\leq D_{\mathrm{h,max}},\,\mathbf{b}_m^T\mathbf{d}_{\mathrm{v,m}}\leq D_{\mathrm{v,max}}, \forall m\in\mathcal{M}.
\end{equation}
In the considered MA-enabled MISO system, the BS energy consumption comprises the energy consumed by the MA drivers, and the energy radiated for data transmission. Thus, the total energy consumption at the BS, $E_{\mathrm{total}}$, is given by 
\begin{equation}\vspace{-0mm}
E_{\mathrm{total}}=\sum_{m\in\mathcal{M}}E_{m,\mathrm{MA}}+\sum_{k\in\mathcal{K}}\|\mathbf{w}_k\|_2^2T_{\mathrm{Data}}.
\end{equation}
Next, we define the energy consumption vector of MA element $m$ $\mathbf{e}_m\in\mathbb{C}^N$ as $\mathbf{e}_m=[e_{m,1},\cdots,e_{m,N}]^T,$
where $e_{m,n}=P_{\mathrm{h,MA}}\frac{|{x}_n-x_m^0|}{v_{\mathrm{h,MA}}}+P_{\mathrm{v,MA}}\frac{|{y}_n-y_m^0|}{v_{\mathrm{v,MA}}}$. Then, the total energy consumption $E_{\mathrm{total}}$ and the average power consumption $\widebar{P}$ at the BS can be rewritten as $E_{\mathrm{total}}=\sum_{m\in\mathcal{M}}\mathbf{b}_m^T\mathbf{e}_m+\sum_{k\in\mathcal{K}}\|\mathbf{w}_k\|^2_2T_{\mathrm{Data}},$ and $\widebar{P}(\mathbf{W},\mathbf{B})=\frac{E_{\mathrm{total}}}{T_{\mathrm{MA}}+T_{\mathrm{Data}}}$, respectively. We aim to minimize the BS's power consumption within one frame while guaranteeing a minimum SINR for each user. The resource allocation problem can be formulated as follows
\begin{eqnarray}
\label{Ori_Problem}
    &&\hspace*{-4mm}\underset{\mathbf{W},\mathbf{B}}{\mino}\hspace*{2mm}\widebar{P}(\mathbf{W},\mathbf{B})\notag\\[-2pt]
    &&\hspace*{-4mm}\mbox{s.t.}\hspace*{4mm} \mbox{C1:}\hspace*{1mm} \mathrm{SINR}_k\geq \gamma_{k},\hspace*{1mm}\forall k\in \mathcal{K},\notag\\[-1pt]
    &&\hspace*{4mm}\mbox{C2:}\hspace*{1mm} \mathbf{b}_m^T\mathbf{D}\mathbf{b}_{m'}\geq D_{\mathrm{min}},\ m\neq m',\ \forall m, m'\in\mathcal{M},\notag\\[-1pt]
    &&\hspace*{4mm}\mbox{C3:}\hspace*{1mm}\mathbf{b}_m^T\mathbf{d}_{\mathrm{h,m}}\leq D_{\mathrm{h,max}},\,\mathbf{b}_m^T\mathbf{d}_{\mathrm{v,m}}\leq D_{\mathrm{v,max}}, \forall m\in\mathcal{M}.\notag\\[-1pt]
    &&\hspace*{4mm}\mbox{C4:}\hspace*{1mm} b_m[n]\in \{0,1\},\hspace*{1mm} \forall n\in \mathcal{N}, \forall m\in \mathcal{M},\notag\\[-1pt]
    &&\hspace*{4mm}\mbox{C5:}\hspace*{1mm} \sum_{n\in\mathcal{N}}b_m[n]=1,\hspace*{1mm}\forall m\in \mathcal{M},
\end{eqnarray}
where $\gamma_k\geq 0$ in constraint C1 denotes the minimum required SINR of user $k$. Problem \eqref{Ori_Problem} is non-convex due to the coupling between $\mathbf{W}$ and $\mathbf{B}$ in constraint C1. In fact, the optimization problem in \eqref{Ori_Problem} is a combinatorial optimization problem which is generally NP-hard\cite{xu2022optimal}. In the following, we develop a BnB-based iterative algorithm to obtain the global optimum of \eqref{Ori_Problem}. \\[-8pt]
{\color{black}
\begin{remark}
    The formulated resource allocation problem for MA-enabled systems in \eqref{Ori_Problem} has a similar form as the antenna selection problems in \cite{sanayei2004antenna,mehanna2013joint}. In particular, by setting the MA movement step size $d$ equal to the antenna spacing and omitting the antenna motion power term in the objective function, the minimum distance constraint C2, and the maximum moving distance constraint C3, the problem in \eqref{Ori_Problem} is equivalent to the antenna selection problems in \cite{sanayei2004antenna, mehanna2013joint}. Therefore, the resource allocation algorithms introduced in Section \rom{3} are also capable of solving antenna selection problems.
\end{remark}}
\subsection{Problem Transformation}
 By introducing an auxiliary matrix $\mathbf{X}=\mathbf{B}\mathbf{W},\ \mathbf{X}\in\mathbb{C}^{N\times K}$, the received signal of user $k$ can be rewritten as follows
\begin{equation}
    y_k=\hat{\mathbf{h}}_k^H\mathbf{X}\mathbf{s}+n_k.
\end{equation}
Thus, we can recast the SINR of user $k$ as
\begin{equation}\label{reform_C1}
    \mathrm{SINR}_k=\frac{|\hat{\mathbf{h}}_k^H\mathbf{x}_k|^2}{\sum_{k'\in\mathcal{K}\setminus\{k\}}|\hat{\mathbf{h}}_k^H\mathbf{x}_{k'}|^2+\sigma_{k}^2},
\end{equation}
where $\mathbf{x}_k$ denotes the $k$-th column of $\mathbf{X}$. Then, by exploiting \eqref{reform_C1} in constraint C1, problem \eqref{Ori_Problem} can be reformulated as\\[-4mm]
\begin{eqnarray}
\label{Ori_Problem1}
    &&\hspace*{-4mm}\underset{\mathbf{B},\,\mathbf{W},\,\mathbf{X}}{\mino}\hspace*{2mm}\widebar{P}(\mathbf{W},\mathbf{B})\notag\\[-2pt]
    &&\hspace*{2mm}\mbox{s.t.}\hspace*{7mm} \mbox{C1},\mbox{C2}, \mbox{C3},\mbox{C4}, \mbox{C5},\mbox{C6:}\hspace*{2mm}\mathbf{X}=\mathbf{B}\mathbf{W}.
\end{eqnarray}
Constraint C6 is non-convex due to the coupling between $\mathbf{B}$ and $\mathbf{W}$. To tackle this non-convexity, we present the following Lemma to reformulate C6 into two convex constraints.
\begin{lemma}
Equality constraint C6 is equivalent to the following linear matrix inequality (LMI) constraints, 
\begin{eqnarray}
\hspace*{-2mm}\mathrm{C6a}\mathrm{:}\hspace*{1mm}\label{sdp}
   \begin{bmatrix}
        \mathbf{U} & \mathbf{X} & \mathbf{B}\\
        \mathbf{X}^H & \mathbf{V} & \mathbf{W}^H\\
        \mathbf{B}^H & \mathbf{W} & \mathbf{I}_M
    \end{bmatrix}&\succeq \mathbf{0},\hspace*{1mm}
\mathrm{C6b}\mathrm{:}\hspace*{1mm}\label{DC}\mathrm{Tr}\big(\mathbf{U}\big)-M\leq 0,\vspace*{-2mm}
\end{eqnarray}
where $\mathbf{U}\in\mathbb{C}^{N\times N}$ and $\mathbf{V}\in\mathbb{C}^{K\times K}$ are two auxiliary optimization variables with $\mathbf{U}\succeq \mathbf{0}$ and $\mathbf{V}\succeq \mathbf{0}$.\\[-4.7mm]
\end{lemma}
\vspace*{-2mm}
\begin{proof} 
Based on \cite{6698281}, C6 is equivalent to C6a and constraint:
\begin{eqnarray}\label{DC2}
\overline{\mbox{C6b}}\mbox{:}&
    \mathrm{Tr}\left(\mathbf{U}-\mathbf{B}\mathbf{B}^H\right)&\leq0.\vspace*{-2mm}
\end{eqnarray}
The left-hand side of \eqref{DC2} can be rewritten as
\begin{equation}\vspace*{-0mm}
         \hspace*{-2mm}\mathrm{Tr}\hspace*{-0.8mm}\left(\mathbf{U}-\mathbf{B}\mathbf{B}^H\right)\hspace*{-0.8mm}\overset{(a)}{=}\hspace*{-0.8mm}\mathrm{Tr}\hspace*{-0.5mm}\left(\mathbf{U}\right)-\hspace*{-1mm}\sum_{m}\hspace*{-1mm}\mathrm{Tr}\hspace*{-0.5mm}\left(\mathbf{b}_m\mathbf{b}_m^H\right)\hspace*{-0.8mm}\overset{(b)}{=}\hspace*{-0.5mm}
         \mathrm{Tr}\hspace*{-0.5mm}\left(\mathbf{U}\right)-M,
\end{equation}
where equalities (a) and (b) hold due to the additivity of the matrix trace and the definition of binary decision vector $\mathbf{b}_m,\forall m$, respectively, which completes our proof.\\[-4mm]\vspace{-1mm}
\end{proof}
On the other hand, minimum distance constraint C2 is also non-convex. Therefore, we exploit the following lemma to convexify constraint C2 based on the binary nature of $\mathbf{B}$.
\begin{lemma}
   Let $\bar{\mathbf{b}}=[\mathbf{b}_1^T,\cdots,\mathbf{b}_M^T]^T$ denote a binary vector collecting all binary selection vectors $\mathbf{b}_m$, $\forall m$. Inequality constraint C2 is equivalent to the following inequality constraint
    \begin{equation}\notag
    \begin{aligned}
         \overline{\mbox{C2}}:&\bar{\mathbf{b}}^T(\frac{1}{2}(-\mathbf{D}_{m,m'}-\mathbf{D}_{m,m'}^T)+\eta\mathbf{I}_{MN})\bar{\mathbf{b}}-\eta M+D_{\mathrm{min}}\\
         &\leq 0,\,m\neq m',\ \forall m, m'\in\mathcal{M},
    \end{aligned}
    \end{equation}
    where $\eta\in\mathbb{R}_+$ is an arbitrary real positive number, $\mathbf{D}_{m,m'}=\hat{\mathbf{I}}_m^T\mathbf{D}\hat{\mathbf{I}}_{m'}$, and $\hat{\mathbf{I}}_m=[\mathbf{0}_{N\times(m-1)N},\mathbf{I}_N,\mathbf{0}_{N\times(M-m)N}]$.\vspace{-1mm}

\end{lemma}
\begin{proof}
 First, we can reformulate constraint C2 as follows,
    \begin{equation}\label{new_C2}
        \hspace*{-2mm}\mbox{C2}\Leftrightarrow\mathbf{b}_m^T(-\mathbf{D})\mathbf{b}_{m'}+D_{\mathrm{min}}\leq 0,\ m\neq m',\ \forall m, m'\in\mathcal{M}.
    \end{equation}
    Next, by defining binary vector $\bar{\mathbf{b}}=[\mathbf{b}_1^T,\cdots,\mathbf{b}_M^T]^T$, we obtain
    \begin{equation}
        \mathbf{b}_m=\hat{\mathbf{I}}_m\bar{\mathbf{b}},
    \end{equation}
    where $\hat{\mathbf{I}}_m=[\mathbf{0}_{N\times(m-1)N},\mathbf{I}_N,\mathbf{0}_{N\times(M-m)N}]$. Then, we can
    rewrite C2 in \eqref{new_C2} equivalently as follows
    \begin{equation}\label{reform_C2}
    \begin{aligned}
        \mbox{C2}&\Leftrightarrow\bar{\mathbf{b}}^T\hat{\mathbf{I}}_m^T(-\mathbf{D})\hat{\mathbf{I}}_{m'}\bar{\mathbf{b}}+D_{\mathrm{min}}\leq 0\\
        &\Leftrightarrow\bar{\mathbf{b}}^T(\frac{1}{2}(-\mathbf{D}_{m,m'}-\mathbf{D}_{m,m'}^T))\bar{\mathbf{b}}+D_{\mathrm{min}}\leq 0,
    \end{aligned}
    \end{equation}
    where $\mathbf{D}_{m,m'}=\hat{\mathbf{I}}_m^T\mathbf{D}\hat{\mathbf{I}}_{m'}$. Moreover, based on the definition of matrix $\mathbf{B}$, the following equality can be obtained.
    \begin{equation}
        \bar{\mathbf{b}}^T\bar{\mathbf{b}}=\mathrm{Tr}(\mathbf{B}^T\mathbf{B})=M.
    \end{equation}
    By adding $\eta\bar{\mathbf{b}}^T\bar{\mathbf{b}}$ to both sides of \eqref{reform_C2}, C2 can be rewritten as 
    \begin{equation}\notag
        \begin{aligned}
            \bar{\mathbf{b}}^T(\frac{1}{2}(-\mathbf{D}_{m,m'}-\mathbf{D}_{m,m'}^T)\bar{\mathbf{b}}+D_{\mathrm{min}}+\eta\bar{\mathbf{b}}^T\bar{\mathbf{b}}&\leq \eta\bar{\mathbf{b}}^T\bar{\mathbf{b}}\\[-2pt]
            \bar{\mathbf{b}}^T(\frac{1}{2}(-\mathbf{D}_{m,m'}-\mathbf{D}_{m,m'}^T+\eta\mathbf{I}_{MN})\bar{\mathbf{b}}+D_{\mathrm{min}}&\leq \eta M\\[-2pt]
            \bar{\mathbf{b}}^T(\frac{1}{2}(-\mathbf{D}_{m,m'}-\mathbf{D}_{m,m'}^T+\eta\mathbf{I}_{MN})\bar{\mathbf{b}}-\eta M+D_{\mathrm{min}}&\leq 0,
        \end{aligned}
    \end{equation}
    which completes the proof\\[-4mm]\vspace{-1mm}
\end{proof}
Note that $ \overline{\mbox{C2}}$ is convex w.r.t. $\mathbf{B}$ for all real values $\eta\geq\lambda_{\mathrm{max}}\left(\frac{1}{2}(\mathbf{D}_{m,m'}+\mathbf{D}_{m,m'}^T)\right)$. In addition, we observe that SINR constraint C1 is non-convex as well. However, if an arbitrary $\mathbf{x}_k$ satisfies the constraints in \eqref{Ori_Problem}, 
$\mathbf{x}_k$ multiplied by an arbitrary phase shift $e^{j\phi}$ also satisfies the constraints without affecting the value of the objective function. Thus, we can leverage the following lemma to transform the non-convex SINR constraint in C1 into two equivalent convex constraints, cf. also \cite{luo2006introduction}.
\begin{lemma}
Without loss of optimality, we assume that $\hat{\mathbf{h}}_k^H\mathbf{x}_k\in\mathbb{R}$. Then, constraint C1 can be equivalently rewritten as
\vspace*{0mm}
\begin{eqnarray}
&& \hspace*{-12mm}{{\mathrm{C1a}}}\mbox{:}\hspace*{1mm}\left({\sum_{k'\neq k}|\hat{\mathbf{h}}_k^H\mathbf{x}_{k'}|^2+\sigma_k^2}\right)^{1/2}\hspace{-2mm}-\frac{\operatorname{Re}\{\hat{\mathbf{h}}_{k}^H\mathbf{x}_k\}}{\sqrt{\gamma_k}}\leq 0, \hspace*{0.5mm}\forall k\in\mathcal{K}\\[-3pt]
&&\hspace*{-12mm}\mathrm{C1b:}\hspace*{1mm}\operatorname{Im}\{\hat{\mathbf{h}}_{k}^H\mathbf{x}_k\}=0,\hspace*{1mm}\forall k\in\mathcal{K}.\vspace*{-2mm}
\end{eqnarray}
C1a and C1b are both convex constraints.\vspace*{-2mm}
\end{lemma}
\vspace*{-0mm}
\begin{proof}
Please refer to \cite[Section \rom{3}]{luo2006introduction}.\\[-4mm]\vspace{-2mm}
\end{proof}
Thus, the resource allocation optimization problem in \eqref{Ori_Problem} can be equivalently reformulated as follows
\begin{eqnarray}
\label{Reform_Problem}
    &&\hspace*{-4mm}\underset{\mathbf{X},\mathbf{W},\mathbf{B},\mathbf{U},\mathbf{V}}{\mino}\hspace*{2mm}\widebar{P}(\mathbf{W},\mathbf{B})\\[-3pt]
    &&\hspace*{2mm}\mbox{s.t.}\hspace*{12mm} \mbox{C1a},\mbox{C1b}, \overline{\mbox{C2}},{\mbox{C3}},\mbox{C4},\mbox{C5},\mbox{C6a},\mbox{C6b}\notag.
\end{eqnarray}
The optimization problem in \eqref{Reform_Problem} is an MINLP problem, where C4 is the only non-convex constraint. Thus, if we employ linear relaxation to constraint C4, the optimization problem becomes a convex optimization problem, and can be optimally solved by numerical solvers such as CVX \cite{grant2008cvx}. On the other hand, for a given fixed binary matrix $\mathbf{B}$, the MINLP problem in \eqref{Reform_Problem} degenerates into a convex optimization problem with respect to variables $\mathbf{X},\mathbf{W},\mathbf{U}$ and $\mathbf{V}$. Thus, rigorous lower and upper bounds of the problem over any given variable subdomain can be generated. In particular, a feasible upper bound is obtained by solving \eqref{Reform_Problem} for a given feasible $\mathbf{B}$, while a lower bound is obtained from relaxations of binary constraint C4. Therefore, the BnB method can be exploited to obtain the optimal solution of the MINLP problem in \eqref{Reform_Problem} \cite{horst2013global,boyd2007branch}.
\subsection{Proposed BnB-based Algorithm}
We develop a BnB-based algorithm to optimally solve the MINLP problem in \eqref{Reform_Problem}. The crux of the BnB method is to divide the feasible region into partitions arranged in a tree structure, with the feasible set of the main problem mapped to the root of the tree. Each node of the BnB tree represents a subset of the feasible set. By subdividing the feasible region into a series of subsets, the problem in \eqref{Reform_Problem} is decomposed into manageable smaller subproblems corresponding to the subsets. In particular, for each subproblem, a lower- and an upper-bound solution is constructed. As the iterative BnB algorithm progresses, based on a predefined node selection strategy, 
one node in the current BnB tree is selected and branched into two new node.
With the expansion of the BnB tree, the feasible set is progressively partitioned into smaller subsets with more accurate lower and upper bounds. Following the above procedure, the gap between the upper bound and the lower bound gradually diminishes with each iteration and the BnB-based algorithm converges towards the global optimum of the considered problem. 
\subsubsection{Initialization:}
\begin{figure}\vspace{-0mm}
    \centering
    \includegraphics[width=1.8in]{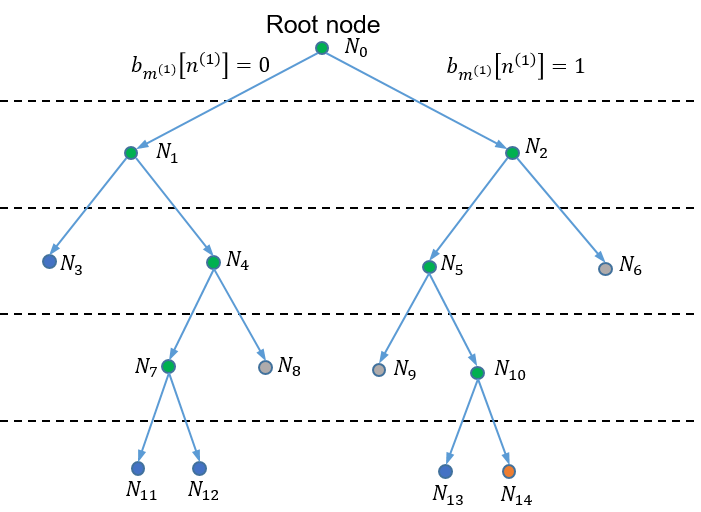}
    \caption{ An illustration of the BnB tree for $M = 2$ and $N = 2$ in the $7$-th iteration. The green, black, grey, and orange dots correspond to internal, external, discarded, and optimal nodes, respectively.}\vspace{2mm}
    \label{fig:BnB}
    \end{figure}
We first initialize the BnB tree $\mathcal{T}^{(0)}=\{\mathcal{N}_0\}$, where $\mathcal{N}_0=\{\mathcal{B}^{(0)},F_L^{(0)},F_U^{(0)}\}$ denotes the root node of $\mathcal{T}^{(0)}$. The search space $\mathcal{B}^{(0)}$ of the binary matrix $\mathbf{B}$ corresponding to the root node of the BnB tree is the Cartesian product of $M\times N$ binary sets, i.e., $\mathcal{B}^{(0)}=\prod_{(m,n)\in\mathcal{L}}\mathcal{B}_{m,n}^{(0)}$, where $\mathcal{L}=\{(1,1),\cdots,(M,N)\}$ denotes the collection of all $(m,n)$ pairs, and $\mathcal{B}_{m,n}^{(0)}=\{0,1\},\forall m,n$, denotes the feasible set of binary variable $b_{m}[n]$. $F_L^{(0)}$ and $F_U^{(0)}$ denote the lower and upper bounds for \eqref{Reform_Problem} in search space $\mathcal{B}^{(0)}$, respectively.

To obtain an initial lower bound for \eqref{Reform_Problem}, we relax $\mathcal{B}^{(0)}$ into a continuous space $\widetilde{\mathcal{B}}^{(0)}=\prod_{(m,n)\in\mathcal{L}}\widetilde{\mathcal{B}}_{m,n}^{(0)}$, where $\widetilde{\mathcal{B}}_{m,n}^{(0)}=\{b_m[n]\,|\,0\leq b_m[n]\leq 1\}$ denotes the relaxed  set of $\mathcal{B}_{m,n}^{(0)}$.
Then, the lower bound $F_L^{(0)}$ for \eqref{Reform_Problem} is obtained by solving 
\begin{eqnarray}
\label{LB_Problem}
    &&\hspace*{-4mm}\underset{\mathbf{X},\mathbf{W},\mathbf{B},\mathbf{U},\mathbf{V}}{\mino}\hspace*{2mm}\widebar{P}(\mathbf{W},\mathbf{B})\\[-4pt]
    &&\hspace*{0mm}\mbox{s.t.}\hspace*{8mm} \mbox{C1a},\mbox{C1b}, \overline{\mbox{C2}},{\mbox{C3}},\mbox{C5},\mbox{C6a},\mbox{C6b},\overline{\mbox{C4}}:\mathbf{B}\in\widetilde{\mathcal{B}}^{(0)},\notag
\end{eqnarray}
which is a convex optimization problem and can be optimally solved by existing numerical convex programming solvers, e.g., CVX \cite{grant2008cvx}. The corresponding optimal solutions of $\mathbf{B}$ and $\mathbf{W}$ are denoted as $\mathbf{B}_L^{(0)}=[\mathbf{b}_{L,1}^{(0)},\cdots,\mathbf{b}_{L,M}^{(0)}]$ and $\mathbf{W}_L^{(0)}=[\mathbf{w}_{L,1}^{(0)},\cdots, \mathbf{w}_{L,K}^{(0)}]$, respectively.
We note that the optimal $\mathbf{B}_L^{(0)}$ may be non-binary since the binary constraint C4 has been relaxed. Thus, the solution of the relaxed problem in \eqref{LB_Problem} establishes a lower bound for \eqref{Reform_Problem}, i.e., $F_L^{(0)}=\widebar{P}(\mathbf{W}_L^{(0)},\mathbf{B}_L^{(0)}).$
On the other hand, we can obtain an initial upper bound for \eqref{Reform_Problem} based on the solution of \eqref{LB_Problem}. In particular, based on $\mathbf{B}_L^{(0)}$, we construct a binary solution $\mathbf{B}_U^{(0)}$ that has minimum Euclidean distance to $\mathbf{B}_L^{(0)}$ by solving the following optimization problem
\begin{eqnarray}\label{B_feasible}
 &&\hspace*{-4mm}\underset{\mathbf{B}}{\mino}\hspace*{2mm}\|\mathbf{B}-\mathbf{B}_L^{(0)}\|_F\\[-3pt]
 &&\hspace*{2mm}\mbox{s.t.}\hspace*{8mm} \overline{\mbox{C2}},{\mbox{C3}},\mbox{C4}, \mbox{C5}\notag,
\end{eqnarray}
which is non-convex due to binary constraint C4. To circumvent this difficulty, we recast C4 equivalently in form of the following inequality constraint and $\overline{\mbox{C4}}$
\begin{equation}\label{SCA_begin}
\begin{aligned}
       {\mbox{C4a}}:&\hspace*{1mm} \sum_{m\in\mathcal{M}}\sum_{n\in\mathcal{N}} (b_m[n]-b_m^2[n])\leq 0, 
\end{aligned}
\end{equation}
where constraint C4a is a difference of convex (DC) functions \cite{le2012exact}. To facilitate low-complexity algorithm design, we resort to the penalty-based method \cite{ng2015secure} to tackle DC constraint C4a and rewrite the problem in \eqref{B_feasible} as follows 
\begin{eqnarray}\label{B_feasible_Penalty}
 &&\hspace*{-2mm}\underset{\mathbf{B}}{\mino}\hspace*{2mm}\|\mathbf{B}-\mathbf{B}_L^{(0)}\|_F+\frac{1}{\mu}\sum_{m\in\mathcal{M}}\sum_{n\in\mathcal{N}} (b_m[n]-b_m^2[n])\notag\\[-3pt]
 &&\hspace*{0mm}\mbox{s.t.}\hspace*{8mm} \overline{\mbox{C2}},{\mbox{C3}},\widebar{\mbox{C4}},\mbox{C4a}, \mbox{C5},
\end{eqnarray}
where $\mu>0$ is a penalty factor for penalizing the violation of the binary constraint. When $\mu$ is sufficiently small, i.e., $\mu\rightarrow 0$, problem \eqref{B_feasible} and problem \eqref{B_feasible_Penalty} are equivalent \cite{le2012exact}. Now, the objective function in \eqref{B_feasible} is in the canonical form of a DC problem in \eqref{B_feasible_Penalty}. Hence, a suboptimal solution of \eqref{B_feasible_Penalty} can be obtained by the SCA method \cite{le2012exact}. In particular, in the $(j+1)$-th iteration of the SCA algorithm, we construct a global underestimator for the term $\sum_{m\in\mathcal{M}}\sum_{n\in\mathcal{N}}b_m^2[n]$ by leveraging first-order Taylor approximation as follows:
\begin{equation}\label{SCA_end}
    \sum_{m\in\mathcal{M}}\sum_{n\in\mathcal{N}}b_m^2[n]\geq \sum_{m\in\mathcal{M}}\sum_{n\in\mathcal{N}} 2b_m^{(j)}[n]b_m[n]-(b_m^{(j)}[n])^2, 
\end{equation}
where $b_m^{(j)}[n]$ is the solution for $b_m[n]$ in the $j$-th iteration. As a result, the optimization problem in the $j$-th iteration of the SCA algorithm is given by
\begin{eqnarray}
\label{SCA_problem_B}
    &&\hspace*{-4mm}\underset{\mathbf{B}}{\mino}\hspace*{2mm}\|\mathbf{B}-\mathbf{B}_L^{(0)}\|_F+\frac{1}{\mu}\sum_{m}\sum_{n}\notag\\ [-2pt] 
    &&\hspace*{10mm}(b_m[n]-2b_m^{(j)}[n]b_m[n]+(b_m^{(j)}[n])^2)\notag\\[-1pt]
    &&\hspace*{0mm}\mbox{s.t.}\hspace*{8mm}\overline{\mbox{C2}},{\mbox{C3}},\overline{\mbox{C4}}, \mbox{C5}.
\end{eqnarray}
The problem in \eqref{SCA_problem_B} is convex and can be optimally solved by CVX. The proposed SCA algorithm to obtain $\mathbf{B}_U^{(0)}$ is summarized in \textbf{Algorithm 1}. The proposed algorithm can obtain a suboptimal binary solution of \eqref{B_feasible}. An upper bound of \eqref{Reform_Problem} can be generated based on any feasible $\mathbf{B}$ satisfying constraints $\overline{\mbox{C2}},{\mbox{C3}},{\mbox{C4}}, \mbox{C5}$ \cite{boyd2007branch}. Thus, based on the binary solution $\mathbf{B}_U^{(0)}$, an upper bound for \eqref{Reform_Problem} is obtained by solving the optimization problem in \eqref{LB_Problem} for fixed $\mathbf{B}=\mathbf{B}_U^{(0)}$, e.g., by using again CVX\cite{grant2008cvx}. Here, $\mathbf{W}_U^{(0)}=[\mathbf{w}_{U,1}^{(0)},\cdots,\mathbf{w}_{U,K}^{(0)}]$ denotes the optimal $\mathbf{W}$ obtained by solving \eqref{LB_Problem} with fixed $\mathbf{B}=\mathbf{B}_U^{(0)}$. Then, an initial upper bound of \eqref{Reform_Problem} is given by $F_U^{(0)}=\widebar{P}(\mathbf{W}_U^{(0)},\mathbf{B}_U^{(0)})$.

\begin{algorithm}[t]
\caption{SCA Algorithm}\vspace{-0mm}
\begin{algorithmic}[1]
\small
\STATE Set iteration index $j=0$ and initialize $\mathbf{B}$ as $\mathbf{B}^{(0)}$. Set convergence tolerance $0<\Delta_{\mathrm{SCA}}\ll 1$ and penalty factor $0<\mu\ll 1$.
\REPEAT
\STATE Set $j=j+1$. Solve \eqref{SCA_problem_B} for given $\mathbf{B}^{(j-1)}$ and update $\mathbf{B}^{(j)}$ as the optimal solution of \eqref{SCA_problem_B}.
\UNTIL $\frac{\|\mathbf{B}^{(j-1)}-\mathbf{B}^{(j)}\|_F}{\|\mathbf{B}^{(j-1)}\|_F}\leq \Delta_{\mathrm{SCA}}$
\end{algorithmic}
\end{algorithm}
\subsubsection{Branch and Bound:}
At the beginning of the $t$-th iteration, the BnB tree obtained in the last iteration is given by $\mathcal{T}^{(t-1)}$. Here, we let sets $\mathcal{I}^{(t-1)}$ and $\mathcal{E}^{(t-1)}$ denote the collections of internal and external nodes of the BnB tree $\mathcal{T}^{(t-1)}$, respectively, as shown in Fig. \ref{fig:BnB}. Here, internal and external nodes represent the nodes with and without child nodes, respectively. 
Then, we select the node $\mathcal{N}_{i_t}=\{\mathcal{B}^{(i_t)},F_L^{(i_t)},F_U^{(i_t)}\}$, $\mathcal{N}_{i_t}\in\mathcal{E}^{(t-1)}$, from the external nodes with the smallest lower bound $F_L^{(i_t)}$, where $i_t$ denotes the index of the selected node. Here, $\mathcal{B}^{(i_t)}=\prod_{(m,n)\in\mathcal{L}}\mathcal{B}_{m,n}^{(i_t)}$ denotes the search space of $\mathbf{B}$ corresponding to the selected node $\mathcal{N}_{i_t}$. In particular, $\mathcal{B}^{(i_t)}$ contains binary sets corresponding to the undetermined binary optimization variables to be optimized in future iterations, i.e.,  $\mathcal{B}_{m,n}^{(i_t)}=\{0,1\}$, and determined binary variables, i.e., $\mathcal{B}_{m,n}^{(i_t)}=\{0\}$ or $\{1\}$. We let set $\mathcal{U}^{(i_t)}$ and $\mathcal{D}^{(i_t)}$ collect the indices of undetermined and determined binary variables, respectively.

Then, we partition the selected node $\mathcal{N}_{i_t}$ into two child nodes based on the Euclidean distance between $\mathbf{B}_L^{(i_t)}$ and $\mathbf{B}_U^{(i_t)}$, which are the solutions of $\mathbf{B}$ corresponding to lower and upper bounds $F_L^{(i_t)}$ and $F_U^{(i_t)}$, respectively.
In particular, we find the undetermined binary variable $b_{m^{(t)}}[n^{(t)}]$ with index $(m^{(t)},n^{(t)})\in\mathcal{U}^{(i_t)}$, where $(m^{(t)},n^{(t)})$ is given by
\begin{equation}\label{Branching}
(m^{(t)},n^{(t)})=\underset{(m,n)\in\mathcal{U}^{(i_t)}}{\mathrm{argmax}}|b_{L,m}^{(i_t)}[n]-b_{U,m}^{(i_t)}[n]|,
\end{equation}
where $b_{L,m}^{(i_t)}[n]$ and $b_{U,m}^{(i_t)}[n]$ denote the $n$-th element in the $m$-th column of $\mathbf{B}_L^{(i_t)}$ and $\mathbf{B}_U^{(i_t)}$, respectively.
Then, the feasible set of $b_{m^{(t)}}[n^{(t)}]$, i.e., $\mathcal{B}_{m^{(t)},n^{(t)}}^{(i_t)}$ is divided into two subsets $(\mathcal{B}_{m^{(t)},n^{(t)}})_0=\{0\}$ and $(\mathcal{B}_{m^{(t)},n^{(t)}})_1=\{1\}$. In the $t$-th iteration, resource allocation problems $\mathcal{P}_i, i\in\{0,1\}$, corresponding to the two child nodes are given by
\begin{eqnarray}
\label{LB_Problem_branch}
    \mathcal{P}_i:&&\hspace*{-4mm}\underset{\mathbf{X},\mathbf{W},\mathbf{B},\mathbf{U},\mathbf{V}}{\mino}\hspace*{2mm}\widebar{P}(\mathbf{W},\mathbf{B})\notag\\[-2pt]
    &&\hspace*{-2mm}\mbox{s.t.}\hspace*{2mm} \mbox{C1a},\mbox{C1b}, \overline{\mbox{C2}},{\mbox{C3}},\mbox{C5},\mbox{C6a},\mbox{C6b},\\[-2pt]
    &&\hspace*{-12mm}\widetilde{\mbox{C4a}}\hspace{-0.5mm}:\hspace{-0.5mm}b_m[n]\hspace{-0.5mm}\in\hspace{-0.5mm}{\mathcal{B}}_{m,n}^{(i_t)},\, \forall (m,n)\hspace{-0.5mm}\neq\hspace{-0.5mm}(m^{(t)},n^{(t)}),\,\widetilde{\mbox{C4b}}\hspace{-0.5mm}:\hspace{-0.5mm}b_{m^{(t)}}[n^{(t)}]\hspace{-1mm}=\hspace{-1mm}i\notag.
\end{eqnarray}
Note that constraint $\widetilde{\mbox{C4a}}$ is non-convex since it contains undetermined binary variables $b_{m}[n]$ to be optimized in future iterations. Here, we relax the feasible set of the undetermined binary variables $b_{m}[n]$ into the corresponding convex hull, i.e., $\widetilde{\mathcal{B}}_{m,n}^{(i_t)}=\{b_{m}[n]\,|\,0\leq b_{m}[n]\leq 1\}$, $\forall(m,n)\in\mathcal{U}^{(i_t)}\setminus(m^{(t)},n^{(t)}) $. Then, the lower bound of \eqref{LB_Problem_branch} is obtained by solving the following relaxed problems $\widetilde{\mathcal{P}}_i$, $i\in\{0,1\}$
\begin{eqnarray}
\label{LB_Problem_branch_LB}
    \widetilde{\mathcal{P}}_i:&&\hspace*{-4mm}\underset{\mathbf{X},\mathbf{W},\mathbf{B},\mathbf{U},\mathbf{V}}{\mino}\hspace*{2mm}\widebar{P}(\mathbf{W},\mathbf{B})\notag\\[-1mm]
    &&\hspace*{-2mm}\mbox{s.t.}\hspace*{2mm} \mbox{C1a},\mbox{C1b}, \overline{\mbox{C2}},{\mbox{C3}},\widetilde{\mbox{C4b}},\mbox{C5},\mbox{C6a},\mbox{C6b}\notag,\\[-0mm]
    &&\hspace*{4mm}\widetilde{\mbox{C4c}}:b_m[n]\in{\mathcal{B}}_{m,n}^{(i_t)},\,
    \forall (m,n)\in\mathcal{D}^{(i_t)},\\[-0mm]
    &&\hspace*{4mm}\widetilde{\mbox{C4d}}: 0 \leq b_{m}[n]\leq 1, \forall (m,n)\in\mathcal{U}^{(i_t)}\setminus(m^{(t)},n^{(t)}).\notag
\end{eqnarray}
Note that $\widetilde{\mathcal{P}}_i$ is a convex optimization problem and can be optimally solved by CVX. Here, the optimal solutions of relaxed problems $\widetilde{\mathcal{P}}_{0}$ and $\widetilde{\mathcal{P}}_{1}$ are denoted as $(\mathbf{B}^{(t)}_{L,0},\mathbf{W}^{(t)}_{L,0})$ and $(\mathbf{B}^{(t)}_{L,1},\mathbf{W}^{(t)}_{L,1})$, respectively. The lower bounds of problems $\widetilde{\mathcal{P}}_{0}$ and $\widetilde{\mathcal{P}}_{1}$ are given by the objective function values $F_{L,0}^{(t)}=\widebar{P}(\mathbf{W}^{(t)}_{L,0},\mathbf{B}^{(t)}_{L,0})$ and $F_{L,1}^{(t)}=\widebar{P}(\mathbf{W}^{(t)}_{L,1},\mathbf{B}^{(t)}_{L,1})$, respectively. 

Moreover, based on optimal solutions $\mathbf{B}^{(t)}_{L,0}$ and $\mathbf{B}^{(t)}_{L,1}$ of the relaxed problems, we can obtain the corresponding binary solution $\mathbf{B}^{(t)}_{U,0}$ and $\mathbf{B}^{(t)}_{U,1}$ based on \textbf{Algorithm 1}. Then, we solve the problem in \eqref{LB_Problem_branch} by fixing the binary variables according to $\mathbf{B}=\mathbf{B}^{(t)}_{U,0}$ and $\mathbf{B}=\mathbf{B}^{(t)}_{U,1}$ and obtain the corresponding upper bounds $F_{U,0}^{(t)}$ and $F^{(t)}_{U,1}$, respectively. After deriving these lower and upper bounds, we expand the BnB tree $\mathcal{T}^{(t-1)}$ by introducing two new child nodes of the selected node $\mathcal{N}_{i_t}$ corresponding to $b_{m^{(t)}}[n^{(t)}]=0$ and $b_{m^{(t)}}[n^{(t)}]=1$, respectively. 
Let $\mathcal{T}^{(t)}$ denote the BnB tree after introducing the new child nodes. Then, the global lower and upper bounds in the $t$-th iteration, i.e., $\mathrm{LB}^{(t)}$ and $\mathrm{UB}^{(t)}$, are given by the smallest lower and upper bound among all external nodes in $\mathcal{T}^{(t)}$.
Next, the search tree is pruned by discarding the nodes whose lower bounds are worse than the current upper bound $\mathrm{UB}^{(t)}$, as shown in Fig. \ref{fig:BnB}.
\setlength{\textfloatsep}{0pt}
\begin{algorithm}[t]
\caption{Proposed BnB-based Algorithm for Perfect CSI}
\begin{algorithmic}[1]
\small
\STATE Solve problem in \eqref{LB_Problem} to obtain optimal solutions $\mathbf{B}_L^{(0)}$ and $\mathbf{W}_L^{(0)}$. Initialize lower bound $\mathrm{LB}=F_L^{(0)}$.
\STATE Compute the binary solution $\mathbf{B}_U^{(0)}$ according to \textbf{Algorithm 1} and obtain the upper bound $\mathrm{UB}=F_U^{(0)}$ by solving problem in \eqref{LB_Problem} for fixed $\mathbf{B}=\mathbf{B}_U^{(0)}$.
\STATE Initialize BnB tree $\mathcal{T}^{(0)}=\{\mathcal{N}_0\}$. Set convergence tolerance $0\leq\Delta_{\mathrm{BnB}}$ and iteration index $t=0$.
\REPEAT
\STATE $t=t+1$,
\STATE Select the external node $\mathcal{N}_{i_t}$ with the smallest lower bound.
\STATE Obtain the branching index $(m^{(t)},n^{(t)})$. Divide feasible set $\mathcal{B}^{(i_t)}$ into two subsets corresponding to $b_{m^{(t)}}[n^{(t)}]=0$ and $b_{m^{(t)}}[n^{(t)}]=1$, respectively.    
\STATE Solve the relaxed version of the two subproblems $\mathcal{P}_{0}$ and $\mathcal{P}_{1}$ in \eqref{LB_Problem_branch}. Store the corresponding lower bounds $F_{L,0}^{(t)}$ and $F_{L,1}^{(t)}$.
\STATE Compute the binary solution $\mathbf{B}_{U,0}^{(t)}$ and $\mathbf{B}_{U,1}^{(t)}$ based on $\mathbf{B}_{L,0}^{(t)}$ and $\mathbf{B}_{L,1}^{(t)}$, respectively, according to \textbf{Algorithm 1}.
\STATE Solve the problem in \eqref{LB_Problem_branch} by fixing $\mathbf{B}=\mathbf{B}_{U,0}^{(t)}$ and $\mathbf{B}=\mathbf{B}_{U,1}^{(t)}$. Store the upper bounds $F_{U,0}^{(t)}$ and $F_{U,1}^{(t)}$.
\STATE Expand the tree by adding the two child nodes corresponding to $b_{m^{(t)}}[n^{(t)}]=0$ and $b_{m^{(t)}}[n^{(t)}]=1$ for selected node $\mathcal{N}_{i_t}$. 
\STATE Update $\mathrm{LB}^{(t)}$ and $\mathrm{UB}^{(t)}$ as the smallest lower bound and upper bound among the external nodes in the tree.
\UNTIL $\mathrm{UB}^{(t)}-\mathrm{LB}^{(t)}\leq \Delta_{\mathrm{BnB}}$
\end{algorithmic}
\end{algorithm}
\subsubsection{Overall Algorithm:}
 The complete BnB procedure is outlined in \textbf{Algorithm 2}. Some further remarks are as follows:\\
       \textit{Optimality and convergence:} The set partitioning, node branching, and bound update steps are iteratively performed to reduce the difference between the lower and upper bounds. According to \cite{horst2013global}, the proposed BnB-based algorithm is guaranteed to converge within a finite number of iterations to an $\epsilon$-optimal solution for a given convergence tolerance. \\
    \textit{Complexity:}\hspace{-1mm} {\color{black} In each iteration of \textbf{Algorithm 2}, we add two new nodes to the BnB tree. For each new node $i$, we have to obtain the corresponding upper and lower bounds. In particular, to determine a lower bound, we have to solve optimization problem $\widetilde{\mathcal{P}}_i$ involving the semidefinite programming (SDP) constraint C6a with dimension $(N+M+K)$. Thus, the computational complexity of solving the optimization problem for the lower bound in each iteration is given by $\mathcal{O}((N+M+K)^3+(N+M+K)^2)$\cite[Theorem 3.12]{bomze2010interior},
    where $\mathcal{O}(\cdot)$ is the big-O notation.
    On the other hand, to obtain the upper bound, we have to first solve the optimization problem in \eqref{B_feasible} to attain a feasible $\mathbf{B}$, and then solve (32) with the obtained $\mathbf{B}$. To solve the problem in \eqref{B_feasible}, we apply the SCA-based \textbf{Algorithm 1}, which iteratively solves the convex quadratic constraint quadratic programming (QCQP) problem in (30) involving $\frac{M\times (M-1)}{2}$ QCQP constraints of dimension $(M\times N)$. 
Based on \cite{boyd2004convex}, the QCQP problem in (30) is first reformulated into an equivalent SDP problem comprising $\frac{M\times (M-1)}{2}$ SDP constraints of dimension $(M\times N+1)$. Thus, the computational complexity of finding a feasible solution with \textbf{Algorithm 1} for each node is given by $\mathcal{O}(I_{\mathrm{itr}}((\frac{M\times (M-1)}{2})(M\times N+1)^3+(\frac{M\times (M-1)}{2})^2(M\times N+1)^2+(M\times N+1)^3))$
, where $I_{\mathrm{itr}}$ denotes the number of SCA iterations\cite[Theorem 3.12]{bomze2010interior}. Moreover, since for a given $\mathbf{B}$ the optimization problem in (32) has a similar form as $\widetilde{\mathcal{P}}_i$ for the lower bound, the computational complexity of solving optimization problem ${\mathcal{P}}_i$ for the upper bound in each iteration is also given by $\mathcal{O}((N+M+K)^3+(N+M+K)^2)$\cite[Theorem 3.12]{bomze2010interior}. Thus, the overall computational complexity for each new node is given by $\mathcal{O}(I_{\mathrm{itr}}((\frac{M\times (M-1)}{2})(M\times N+1)^3+(\frac{M\times (M-1)}{2})^2(M\times N+1)^2+(M\times N+1)^3)+2(N+M+K)^3+2(N+M+K)^2)$.}
    {\color{black} Furthermore, we note that our simulation results in Section \rom{5} demonstrate that the proposed BnB-based algorithm requires significantly fewer iterations to converge compared to an exhaustive search (ES) \cite{horst2013global}, even though the worst-case computational complexity scales exponentially with the number of MA elements \cite{horst2013global} since the optimization problem in \eqref{Reform_Problem} is a typical NP-hard problem.}

\subsection{Proposed SCA-based Algorithm}
To strike a balance between complexity and optimality, we develop a suboptimal SCA-based algorithm with reduced complexity. Following a similar procedure as in Section \rom{3}-C for deriving the initial upper bound, we recast binary constraint C4 as the two convex constraints, $\overline{\mbox{C4}}$ and C4a. Then, we resort to the penalty method \cite{ng2015secure} to tackle DC constraint C4a and rewrite the problem in \eqref{Reform_Problem} as follows 
\begin{eqnarray}
\label{Penalty_problem_bnb}
    &&\hspace*{-6mm}\underset{\mathbf{X},\mathbf{W},\mathbf{B},\mathbf{U},\mathbf{V}}{\mino}\hspace*{1mm}\widebar{P}(\mathbf{W},\mathbf{B})+\frac{1}{\mu}\sum_{m,n}(b_m[n]-b_m^2[n])\\[-1pt]
    &&\hspace*{0mm}\mbox{s.t.}\hspace*{12mm} \mbox{C1a},\mbox{C1b}, \overline{\mbox{C2}},{\mbox{C3}},\overline{\mbox{C4}},\mbox{C5},\mbox{C6a},\mbox{C6b}\notag,
\end{eqnarray}
where $\mu>0$ is a sufficiently small penalty factor such that problem \eqref{Reform_Problem} and problem \eqref{Penalty_problem_bnb} are equivalent \cite{le2012exact}. 
By applying the SCA technique, we formulate the optimization problem in the $(j+1)$-th iteration of the SCA algorithm as
\begin{eqnarray}
\label{SCA_problem_BnB}
    &&\hspace*{-8mm}\underset{\substack{\mathbf{X},\mathbf{W},\mathbf{B},\\\mathbf{U},\mathbf{V}}}{\mino}\hspace*{1mm}\widebar{P}(\mathbf{W},\hspace*{-0mm}\mathbf{B})\hspace*{-0.5mm}+\hspace*{-0.5mm}\frac{1}{\mu}\hspace*{-0.5mm}\sum_{m,n}\hspace*{-0.5mm}b_m[n]\hspace*{-0.5mm}-\hspace*{-0.5mm}2b_m^{(j)}[n]b_m[n]\hspace*{-0.5mm}+\hspace*{-0.5mm}(b_m^{(j)}[n])^2\hspace*{-2mm}\notag\\[-0pt]
    &&\hspace*{0mm}\mbox{s.t.}\hspace*{8mm}\mbox{C1a},\mbox{C1b}, \overline{\mbox{C2}},{\mbox{C3}},\overline{\mbox{C4}},\mbox{C5},\mbox{C6a},\mbox{C6b},
\end{eqnarray}
where $b_m^{(j)}[n]$ is the solution for $b_m[n]$ in the $j$-th iteration. 
Note that problem \eqref{SCA_problem_BnB} is convex and can be optimally solved. 
The proposed iterative SCA algorithm is summarized in \textbf{Algorithm 3}\footnote{\color{black}We note that the proposed BnB-based and SCA-based algorithms are not limited to the field response channel model. The proposed algorithms can be easily extended to other multipath channel models, e.g., Rician channels in \cite{hu2024movable,hu2024two} by slightly modifying the definition of the effective channel vectors.}. Some further remarks are given as follows:\\
   \textit{Initial point:} A matrix, $\mathbf{B}$, satisfying constraints $\overline{\mbox{C2}},{\mbox{C3}},\overline{\mbox{C4}},$ and C5 is randomly generated.\\
 \textit{Optimality and convergence:} 
    The solution of problem \eqref{SCA_problem_BnB} provides an upper limit for the problem in \eqref{Reform_Problem}. Through solving \eqref{SCA_problem_BnB} iteratively, we progressively tighten this upper limit. Note that the proposed suboptimal algorithm converges to a locally optimal solution of \eqref{Reform_Problem} \cite{razaviyayn2013unified}.\\
    \textit{Complexity:} In each iteration, we only need to solve one convex programming problem. Moreover, \textbf{Algorithm 3} is guaranteed to converge in a polynomial number of iterations\cite{razaviyayn2013unified}. Thus, \textbf{Algorithm 3} exhibits polynomial computational complexity. Given that the problem in \eqref{SCA_problem_BnB} involves $K$ second-order cone (SOC) constraints involving optimization variable $\mathbf{X}\in\mathbb{C}^{N\times K}$ with a cone dimension of $\mathcal{O}(K)$, the computational complexity of each iteration of the proposed SCA-based algorithm is given by $\mathcal{O}(\log\frac{1}{\rho}(K^{\frac{7}{2}}N^3+K^{\frac{7}{2}}N^2+K^{\frac{7}{2}}))$\cite[Theorem 3.11]{bomze2010interior}.\\[-12pt]
\begin{algorithm}[t]
\caption{Proposed SCA-based Algorithm for Perfect CSI}
\begin{algorithmic}[1]
\small
\STATE Set iteration index $j=0$ and generate a feasible $\mathbf{B}^{(0)}$. Set $0<\Delta_{\mathrm{SCA}}\ll 1$ and penalty factor $0<\mu\ll 1$
\REPEAT
\STATE Set $j=j+1$. Solve \eqref{SCA_problem_BnB} for a given $\mathbf{B}^{(i-1)}$ and update $\mathbf{B}^{(i)}$ as the optimal solution $\mathbf{B}$ of \eqref{SCA_problem_BnB}
\UNTIL $\frac{\|\mathbf{B}^{(j-1)}-\mathbf{B}^{(j)}\|_F}{\|\mathbf{B}^{(j-1)}\|_F}\leq \Delta_{\mathrm{SCA}}$
\end{algorithmic}
\end{algorithm}
\section{Algorithm Design for Imperfect CSI}\vspace*{-0mm}
In this section, we reformulate the problem in \eqref{Ori_Problem1} accounting for imperfect CSI based on the norm-bounded error model in \eqref{CSI_model}. 
Accordingly, we suitably modify the BnB-based and SCA-based methods proposed in Section \rom{3}, respectively, to tackle the new robust resource allocation problem.\vspace{-1mm}
\subsection{Problem Reformulation}


Considering the problem in \eqref{Ori_Problem}, the robust BS beamforming and MA position optimization problem is formulated as follows
\begin{eqnarray}
\label{Ori_Problem_robust}
    &&\hspace*{-6mm}\underset{\mathbf{W},\mathbf{B}}{\mino}\hspace*{2mm}\widebar{P}(\mathbf{W}, \mathbf{B})\notag\\[-3pt]
    &&\hspace*{-6mm}\mbox{s.t.}\hspace*{1mm} \widehat{\mbox{C1}}:\hspace*{1mm} \min_{\Delta{\bm{\Psi}}_k\in\Omega_{\Delta{\bm{\Psi}}_k}}\frac{|\hat{\mathbf{h}}_{k}^H\mathbf{B}\mathbf{w}_k|^2}{\sum_{k'\neq k}|\hat{\mathbf{h}}_{k}^H\mathbf{B}\mathbf{w}_{k'}|^2+\sigma_k^2}\geq \gamma_{k},\hspace*{1mm}\forall k\in \mathcal{K},\notag\\[-2pt]
    &&\hspace*{-2mm}\overline{\mbox{C2}},\mbox{C3},\mbox{C4},\mbox{C5},
\end{eqnarray}
where constraint $\widehat{\mbox{C1}}$ includes infinitely many non-convex inequality constraints due to the continuity of the CSI uncertainty sets, and cannot be reformulated to a convex constraint by Lemma 3. Instead, the numerator of the SINR constraint $\widehat{\mbox{C1}}$ can be rewritten as $|\hat{\mathbf{h}}_{k}^H\mathbf{B}\mathbf{w}_k|^2=\hat{\mathbf{h}}_{k}^H{\mathbf{B}}\mathbf{W}_k{\mathbf{B}}^H\hat{\mathbf{h}}_{k},$
where $\mathbf{W}_k=\mathbf{w}_k\mathbf{w}_k^H$. Also, the denominator can be rewritten in a similar manner, and we can recast constraint $\widehat{\mbox{C1}}$ as follows
\begin{equation}
\begin{aligned}
    \widehat{\mbox{C1}}\Leftrightarrow \min_{\hat{\mathbf{h}}_k\in\Omega_{\hat{\mathbf{h}}_k}}\sigma_k^2\gamma_k+\hat{\mathbf{h}}_{k}^H{\mathbf{B}}\left(\gamma_k\sum_{k'\neq k}{\mathbf{W}}_{k'}-{\mathbf{W}}_{k}\right){\mathbf{B}}^H\hat{\mathbf{h}}_{k}\leq 0.\notag
\end{aligned}
\end{equation}
By substituting $\hat{\mathbf{h}}_k=\mathbf{G}_k^H(\bar{\bm{\psi}}_k+\Delta\bm{\psi}_k)$,  we can rewrite $\widehat{\mbox{C1}}$ as
\begin{equation}
\begin{aligned}
&{\bar{\bm{\psi}}}_{k}^H\mathbf{G}_k{\mathbf{B}}\widetilde{\mathbf{W}}_k{\mathbf{B}}^H\mathbf{G}^H_k{\bar{\bm{\psi}}}_k\hspace{-1mm}+\hspace{-1mm}2\mathrm{Re}\{{\bar{\bm{\psi}}}_{k}^H\mathbf{G}_k{\mathbf{B}}\widetilde{\mathbf{W}}_k{\mathbf{B}}^H\mathbf{G}_k^H\Delta{\bm{\psi}}_{k}\}\\[-2pt]
&+\Delta{{\bm{\psi}}}_k^H\mathbf{G}_k{\mathbf{B}}\widetilde{\mathbf{W}}_k{\mathbf{B}}^H\mathbf{G}_k^H\Delta{{\bm{\psi}}}_k+\sigma_k^2\gamma_k\leq 0,\,\forall\|\Delta {\bm{\psi}}_{k}\|_2^2\leq \epsilon_k^2,\notag 
\end{aligned}
\end{equation}
where $\widetilde{\mathbf{W}}_k=\gamma_k\sum_{k'\in\mathcal{K}\setminus\{k\}}{\mathbf{W}}_{k'}-{\mathbf{W}}_k$. Here, constraint $\widehat{\mbox{C1}}$ consists of infinitely many inequality constraints due to the continuity of the CSI uncertainty sets. To facilitate algorithm design, we apply following lemma to convexify constraint $\widehat{\mbox{C1}}$.
\begin{lemma}\label{S_procedure}
    (Generalized S-Lemma \cite{boyd2004convex}): Let $f_i(\mathbf{x}),$ $\forall i\in\{0,\dots,I\}$, be a real-valued function of vector $\mathbf{x}\in\mathbb{C}^{N\times 1}$ and be defined as $f_i(\mathbf{x})=\mathbf{x}^H\mathbf{A}_i\mathbf{x}+2\operatorname{Re}\{\mathbf{a}_i^H\mathbf{x}\}+a_i,$
    where $\mathbf{A}_i\in\mathbb{S}^N$, $\mathbf{a}_i\in\mathbb{C}^{N\times 1}$, and $a_i\in\mathbb{R}$, $\forall i\in\{0,\cdots,I\}$. Then, the implication $f_i(\mathbf{x})\leq 0,\, \forall i\in\{1,\cdots,I\}\Rightarrow f_0(\mathbf{x})\leq 0$ holds if and only if there exist some real numbers $\lambda_i\geq 0$, $\forall i\in\{1,\cdots,I\}$, such that
    \begin{equation}
        \sum_{i=1}^I\lambda_i
    \begin{bmatrix}
        \mathbf{A}_i & \mathbf{a}_i \\
        \mathbf{a}_i^H & a_i
    \end{bmatrix}-
     \begin{bmatrix}
        \mathbf{A}_0 & \mathbf{a}_0 \\
        \mathbf{a}_0^H & a_0
    \end{bmatrix}\succeq \mathbf{0}.
    \end{equation}\vspace{-2mm}
\end{lemma}
{\color{black} Here, we define $f_0(\Delta {\bm{\psi}}_{k})=\Delta{{\bm{\psi}}}_k^H\mathbf{A}_0\Delta{{\bm{\psi}}}_k+2\operatorname{Re}\{\mathbf{a}_0^H\Delta{{\bm{\psi}}}_k\}+a_0\leq 0$ and $f_1(\Delta {\bm{\psi}}_{k})=\Delta {\bm{\psi}}_{k}^H\mathbf{A}_1\Delta {\bm{\psi}}_{k}+a_1$, where $\mathbf{A}_0=\mathbf{G}_k{\mathbf{B}}\widetilde{\mathbf{W}}_k{\mathbf{B}}^H\mathbf{G}^H_k$, $\mathbf{a}_0^H={\bar{\bm{\psi}}}_{k}^H\mathbf{G}_k{\mathbf{B}}\widetilde{\mathbf{W}}_k{\mathbf{B}}^H\mathbf{G}_k^H$, \\$a_0={\bar{\bm{\psi}}}_{k}^H\mathbf{G}_k{\mathbf{B}}\widetilde{\mathbf{W}}_k{\mathbf{B}}^H\mathbf{G}^H_k{\bar{\bm{\psi}}}_k+\sigma_k^2\gamma_k$, $\mathbf{A}_1=\mathbf{I}_{L_k}$, and $\mathbf{a}_1=-\epsilon_k^2$.  Thus, by leveraging Lemma 1, SINR constraint $\widehat{\mbox{C1}}$ can be rewritten equivalently as follows
\begin{equation}\label{LMI_SINR}
\begin{aligned}
    &\hspace{-4mm}q_k\begin{bmatrix}
       \mathbf{I}_{L_k} & \mathbf{0}_{L_k} \\
        \mathbf{0}_{L_k}^H &-\epsilon_k^2
    \end{bmatrix}-\\
&\hspace{-4mm}\begin{bmatrix}
       \mathbf{G}_k{\mathbf{B}}\widetilde{\mathbf{W}}_k{\mathbf{B}}^H\mathbf{G}^H_k & \mathbf{G}_k{\mathbf{B}}\widetilde{\mathbf{W}}_k{\mathbf{B}}^H\mathbf{G}_k^H{\bar{\bm{\psi}}}_{k} \\
        {\bar{\bm{\psi}}}_{k}^H\mathbf{G}_k{\mathbf{B}}\widetilde{\mathbf{W}}_k{\mathbf{B}}^H\mathbf{G}_k^H &\hspace{-2mm}{\bar{\bm{\psi}}}_{k}^H\mathbf{G}_k{\mathbf{B}}\widetilde{\mathbf{W}}_k{\mathbf{B}}^H\mathbf{G}^H_k{\bar{\bm{\psi}}}_k+\sigma_k^2\gamma_k
    \end{bmatrix}\succeq \mathbf{0},\\
    \Leftrightarrow&\mathbf{P}_k-
\widetilde{\mathbf{G}}_k^H{\mathbf{B}}\widetilde{\mathbf{W}}_k{\mathbf{B}}^H\widetilde{\mathbf{G}}_k\succeq \mathbf{0}, \forall k,
\end{aligned}    
\end{equation}
where 
\begin{equation}
    \begin{aligned}
        \mathbf{P}_k=
    \begin{bmatrix}
        q_k\mathbf{I}_{L_k} & \mathbf{0} \\
        \mathbf{0} & -q_{k}\epsilon_{k}^2-\sigma_k^2\gamma_k
    \end{bmatrix},\,\widetilde{\mathbf{G}}_k=[\mathbf{G}_{k}\quad \mathbf{G}_{k}{\bar{\bm{\psi}}}_{k}^H],
    \end{aligned}
\end{equation}
and $q_k\geq 0$ is an auxiliary optimization variable.}
We note that SINR constraint $\widehat{\mbox{C1}}$ is non-convex w.r.t. $\mathbf{B}$ and ${\mathbf{W}}_k$ due to the coupling between $\mathbf{B}$ and ${\mathbf{W}}_k$. Next, we define a new optimization variable $\hat{\mathbf{X}}_k={\mathbf{B}}{\mathbf{W}}_k{\mathbf{B}}^H$ to recast $\eqref{LMI_SINR}$ as follows
\begin{equation}
    \widehat{\overline{\mbox{C1}}}:\hspace*{1mm} 
    \mathbf{P}_k-
    \widetilde{\mathbf{G}}_k^H\widetilde{\mathbf{X}}_k\widetilde{\mathbf{G}}_k\succeq \mathbf{0}, \forall k,
\end{equation}
where $\widetilde{\mathbf{X}}_k=\gamma_k\sum_{k'\in\mathcal{K}\setminus\{k\}}\hat{\mathbf{X}}_{k'}-\hat{\mathbf{X}}_k$. Note that constraint $\widehat{\overline{\mbox{C1}}}$ is convex w.r.t. variables $\hat{\mathbf{X}}_{k}, \forall k$. For the sake of notational simplicity, we define $\hat{\mathbf{X}}=[\hat{\mathbf{X}}_1,\cdots,\hat{\mathbf{X}}_K]$ to collect all $\hat{\mathbf{X}}_k$, $\forall k$. Then,
\eqref{Ori_Problem_robust} can be recast as the following optimization problem:
\begin{eqnarray}
\label{Ori_Problem_robust_recast}
    &&\hspace*{-10mm}\underset{{\mathbf{W}},\mathbf{B},\hat{\mathbf{X}}}{\mino}\hspace*{2mm}\widebar{P}(\mathbf{W},\mathbf{B})\notag\\[-4pt]
    &&\hspace*{-0mm}\mbox{s.t.}\hspace*{2mm} \widehat{\overline{\mbox{C1}}},\overline{\mbox{C2}},\mbox{C3-C5},{\mbox{C7}}:\hspace*{0mm} \hat{\mathbf{X}}_k={\mathbf{B}}{\mathbf{W}}_k{\mathbf{B}}^H,\hspace*{1mm}\forall k\in \mathcal{K}.
\end{eqnarray}
Note that C7 is a non-convex constraint due to the coupled variables $\mathbf{W}_k$ and ${\mathbf{B}}$. In particular, C7 is neither bilinear nor biconvex. However, thanks to the binary nature of ${\mathbf{B}}$, we are able to convexify C7 by exploiting the following lemma.
\begin{lemma}
    Equality constraint C7 is equivalent to the following LMI constraints, 
    \begin{eqnarray}\label{sdp_robust}
&&\hspace*{-6mm}{\mathrm{C7a}}\mbox{:}\hspace*{-1mm}
   \begin{bmatrix}
        \hat{\mathbf{S}}_k & \hat{\mathbf{X}}_k & {\mathbf{B}}\\
        \hat{\mathbf{X}}_k^H & \hat{\mathbf{T}}_k & \mathbf{Y}_k\\
        {\mathbf{B}}^H & \mathbf{Y}_k^H & \mathbf{I}_M
    \end{bmatrix}\succeq \mathbf{0},\hspace*{1mm}{{{\mathrm{C7c}}}}\mbox{:}\hspace*{-1mm}  \begin{bmatrix}
        \mathbf{U}_k & {\mathbf{Y}}_k & {\mathbf{B}}\\
        {\mathbf{Y}}_k^H & \mathbf{V}_k & \mathbf{W}_k\\
        {\mathbf{B}}^H & \mathbf{W}_k^H & \mathbf{I}_M
    \end{bmatrix}\succeq \mathbf{0},\notag\\[-2pt]
&&\hspace*{-6mm}{{{\mathrm{C7b}}}}\mbox{:}\hspace{1mm}\mathrm{Tr}\big(\hat{\mathbf{S}}_k\big)-M\leq 0,
    \hspace*{3mm}{{{\mathrm{C7d}}}}\mbox{:}\hspace{1mm}\mathrm{Tr}\big({\mathbf{U}}_k\big)-M\leq 0,
\end{eqnarray}
where $\hat{\mathbf{S}}_k\in\mathbb{C}^{N\times N}$, $\hat{\mathbf{T}}_k\in\mathbb{C}^{N\times N}$, $\mathbf{Y}_k\in\mathbb{C}^{N\times M}$, $\mathbf{U}_k\in\mathbb{C}^{N\times N}$, and $\mathbf{V}_k\in\mathbb{C}^{M\times M}$ are the auxiliary variables.\vspace{-2mm}
\end{lemma}
\begin{proof}
        We first define an auxiliary variable $\mathbf{Y}_k$ as $\mathbf{Y}_k={\mathbf{B}}{\mathbf{W}}_k^H$. Then, ${{\mbox{C7}}}$ can be recast into the following constraints
    \begin{equation}
        {{\mbox{C7}}}\Leftrightarrow\hat{\mathbf{X}}_k={\mathbf{B}}{\mathbf{Y}}_k^H,\quad\mathbf{Y}_k= {\mathbf{B}}{\mathbf{W}}_k^H.
    \end{equation}
    Based on Lemma 1, $\hat{\mathbf{X}}_k={\mathbf{B}}{\mathbf{Y}}_k^H$ can be rewritten as ${{\mbox{C7a}}}$ and
    \begin{equation}\label{C7b_original}
        \mathrm{Tr}\big(\hat{\mathbf{S}}_k-{\mathbf{B}}{\mathbf{B}}^H\big)\leq 0.
    \end{equation}
    Because of the binary nature of ${\mathbf{B}}$, we can rewrite \eqref{C7b_original} as the following linear inequality constraint
\begin{equation}\label{C7b}
\begin{aligned}
    \mathrm{Tr}\big(\widehat{\mathbf{S}}_k-{\mathbf{B}}{\mathbf{B}}^H\big)
    &\Leftrightarrow{{\mbox{C7b}}}:\mathrm{Tr}\big(\widehat{\mathbf{S}}_k\big)-M\leq 0.
\end{aligned}
\end{equation}
Similar to the above approach, we can recast constraint $\mathbf{Y}_k= {\mathbf{W}}_k{\mathbf{B}}^H$ as ${{\mbox{C7c}}}$ and ${{\mbox{C7d}}}$,  which completes the proof.\\[-4mm]\vspace*{-2mm}
\end{proof}
For convenience, we define $\hat{\mathbf{S}}=[\hat{\mathbf{S}}_1,\cdots,\hat{\mathbf{S}}_K]$, $\hat{\mathbf{T}}=[\hat{\mathbf{T}}_1,\cdots,\hat{\mathbf{T}}_K]$, $\mathbf{U}=[\mathbf{U}_1,\cdots,\mathbf{U}_K]$, $\mathbf{V}=[\mathbf{V}_1,\cdots,\mathbf{V}_K]$, $\mathbf{Y}=[\mathbf{Y}_1,\cdots,\mathbf{Y}_K]$, $\mathbf{q}=[q_1,\cdots,q_K]^T$, and $\hat{\mathbf{W}}=[\mathbf{W}_1,\cdots,\mathbf{W}_K]$. Moreover, $\bm{\Gamma}=\{\hat{\mathbf{W}},\mathbf{q},\hat{\mathbf{X}},\mathbf{Y},\hat{\mathbf{S}},\hat{\mathbf{T}},\mathbf{V},\mathbf{U}\}$ collects all variables except $\mathbf{B}$. Then, the proposed robust resource allocation problem in \eqref{Ori_Problem_robust_recast} can be recast as follows
\begin{eqnarray}
\label{Reform_Problem_robust}
    &&\hspace*{-8mm}\underset{\substack{\bm{\Gamma},\mathbf{B}}}{\mino}\hspace*{2mm}\widetilde{P}(\hat{\mathbf{W}},\mathbf{B})=\frac{\sum_{m\in\mathcal{M}}\mathbf{b}_m^T\mathbf{e}_m+\sum_{k\in\mathcal{K}}\mathrm{Tr}(\mathbf{W}_k)T_{\mathrm{Data}}}{T_{\mathrm{MA}}+T_{\mathrm{Data}}}\notag\\[-2pt]
    &&\hspace*{-6mm}\mbox{s.t.}\hspace*{7mm} \widehat{\overline{\mbox{C1}}},\overline{\mbox{C2}},\mbox{C3},\mbox{C4},\mbox{C5},\mbox{C7a},\mbox{C7b},\mbox{C7c},\mbox{C7d},\\[-1pt]
    &&\hspace*{5mm}\mbox{C8:}\hspace*{1mm} \mathrm{rank}(\mathbf{W}_k)=1,\hspace*{1mm}\forall k\in \mathcal{K},\,\mbox{C9:}\hspace*{1mm} \mathbf{W}_k\succeq \mathbf{0},\hspace*{1mm}\forall k\in \mathcal{K}.\notag
\end{eqnarray}
The problem in \eqref{Reform_Problem_robust} is non-convex due to binary constraint C4 and rank-one constraint C8. In the following subsection, we introduce a series of transformations to convexify non-convex constraints C4 and C8 and extend the proposed BnB-based algorithm for perfect CSI to the case of imperfect CSI.
\subsection{Proposed BnB-based Algorithm }
We first initialize the BnB tree $\hat{\mathcal{T}}^{(0)}$ for imperfect CSI. Here, the initial search tree $\hat{\mathcal{T}}^{(0)}$ contains one root node $\hat{\mathcal{N}}_0=\{\hat{\mathcal{B}}^{(0)},\hat{F}_L^{(0)},\hat{F}_U^{(0)}\}$, where  $\hat{\mathcal{B}}^{(0)}$, $\hat{F}_L^{(0)}$, and $\hat{F}_U^{(0)}$ denote the initial search space and lower and upper bounds of the problem in \eqref{Reform_Problem_robust}, respectively. Here, $\hat{\mathcal{B}}^{(0)}=\prod_{(m,n)\in\mathcal{L}}\hat{\mathcal{B}}_{m,n}^{(0)}$ is the Cartesian product of $M\times N$ binary sets, where $\hat{\mathcal{B}}_{m,n}^{(0)}=\{0,1\}$, denotes the initial feasible set of $b_{m}[n]$. 

We then derive the initial lower and upper bounds for the considered MA-enabled MIMO systems. In particular, to facilitate the derivation of the lower bound, we relax the initial binary search space to convex set $\widebar{\mathcal{B}}^{(0)}=\prod_{(m,n)\in\mathcal{L}}\widebar{\mathcal{B}}_{m,n}^{(0)}$ with $\widebar{\mathcal{B}}_{m,n}^{(0)}=\{b_m[n]\,|\,0\leq b_m[n]\leq 1\},\forall m,n$. Moreover, we tackle rank-one constraint C8 by SDR \cite{yu2021robust}. Thus, a lower bound of \eqref{Reform_Problem_robust} is obtained by solving the following problem
\begin{eqnarray}
\label{LB_Problem_robust}
    &&\hspace*{-4mm}\underset{\substack{\bm{\Gamma},\mathbf{B}}}{\mino}\hspace*{2mm}\widetilde{P}(\hat{\mathbf{W}},\mathbf{B})\notag\\[-3pt]
    &&\hspace*{2mm}\mbox{s.t.}\hspace*{3mm} \widehat{\overline{\mbox{C1}}},\overline{\mbox{C2}},\mbox{C3},\mbox{C5},\mbox{C7a-C7d},\mbox{C9},\overline{\mbox{C4}}:\hspace*{1mm} \mathbf{B}\in\widebar{\mathcal{B}}^{(0)}.
\end{eqnarray}
The tightness of the relaxation is revealed in following theorem. 
\begin{theorem}
    If problem \eqref{LB_Problem_robust} is feasible, there always exists an optimal beamforming matrix $\mathbf{W}_k^{\mathrm{opt}}$ with $\mathrm{rank}(\mathbf{W}_k^{\mathrm{opt}})\leq 1, \forall k$. \\[-3mm]
\end{theorem}
\vspace{-3mm}
\begin{proof}
Please refer to Appendix D of the extended version of the present paper \cite{wu2024globally}, which will be added in the final version of the paper..\\[-4mm]\vspace{-1mm}
\end{proof}
Note that the problem in \eqref{LB_Problem_robust} is convex w.r.t. $\bm{\Gamma}$ and $\mathbf{B}$. Here, let $\bm{\Gamma}^{(0)}_L$ and $\hat{\mathbf{B}}^{(0)}_L$ denote the optimal solutions of $\bm{\Gamma}$ and $\mathbf{B}$ for optimization problem \eqref{LB_Problem_robust}, respectively. Then, the initial lower bound $\hat{F}_L^{(0)}$ is given by $\hat{F}_L^{(0)}=\widetilde{P}(\hat{\mathbf{W}}^{(0)}_L,\hat{\mathbf{B}}^{(0)}_L).$

On the other hand, we can obtain an upper bound of \eqref{Reform_Problem_robust} based on $\hat{\mathbf{B}}_L^{(0)}$. Following a similar procedure as in \textbf{Algorithm 1}, we can construct a binary solution $\hat{\mathbf{B}}_U^{(0)}=[\hat{\mathbf{b}}_{U,1}^{(0)},\cdots,\hat{\mathbf{b}}_{U,M}^{(0)}]$ based on $\hat{\mathbf{B}}_L^{(0)}$ by minimizing the Euclidean distance between $\hat{\mathbf{B}}_L^{(0)}$ and $\hat{\mathbf{B}}_U^{(0)}$. Then, an upper bound for \eqref{Reform_Problem_robust} is obtained by solving the following optimization problem
\begin{eqnarray}
\label{UB_Problem_robust}
    &&\hspace*{-4mm}\underset{\bm{\Gamma},\mathbf{B}=\hat{\mathbf{B}}_U^{(0)}}{\mino}\hspace*{2mm}\widetilde{P}(\hat{\mathbf{W}},\mathbf{B})\notag\\[-6pt]
    &&\hspace*{2mm}\mbox{s.t.}\hspace*{7mm} \widehat{\overline{\mbox{C1}}},\mbox{C7a},\mbox{C7b},\mbox{C7c},\mbox{C7d},\mbox{C9}.
\end{eqnarray}
Here, although the rank-one constraint is relaxed, the proof of the tightness of the rank-one relaxation is similar to the proof of Theorem 1.
Note that the optimization problem in \eqref{UB_Problem_robust} is convex w.r.t. the variables in $\bm{\Gamma}$. Thus, the optimal solution of \eqref{UB_Problem_robust} can be obtained by CVX. Here, we denote the obtained optimal solution of $\hat{\mathbf{W}}$ as $\hat{\mathbf{W}}^{(0)}_U=[\hat{\mathbf{W}}_{U,1}^{(0)},\cdots,\hat{\mathbf{W}}_{U,K}^{(0)}]$ and the initial upper bound $\hat{F}_U^{(0)}$ is given by $\hat{F}_U^{(0)}=\widetilde{P}(\hat{\mathbf{W}}^{(0)}_U,\hat{\mathbf{B}}_U^{(0)}).$

In the $t$-th iteration of the proposed BnB-based algorithm, we denote $\hat{\mathcal{T}}^{(t-1)}$ as the BnB tree obtained in the last iteration. Here, we let sets $\hat{\mathcal{I}}^{(t-1)}$ and $\hat{\mathcal{E}}^{(t-1)}$ denote the collections of the internal and external nodes of $\hat{\mathcal{T}}^{(t-1)}$. Then, we select the node $\hat{\mathcal{N}}_{\hat{i}_t}=\{\hat{\mathcal{B}}^{(\hat{i}_t)},\hat{F}_L^{(\hat{i}_t)},\hat{F}_U^{(\hat{i}_t)}\}$
from the external nodes of BnB tree $\hat{\mathcal{T}}^{(t-1)}$ with the smallest lower bound $\hat{F}_L^{(\hat{i}_t)}$, where $\hat{i}_t$ denotes the index of the selected node from $\hat{\mathcal{T}}^{(t-1)}$ at the $t$-th iteration. Moreover, $\hat{\mathcal{B}}^{(\hat{i}_t)}=\prod_{(m,n)\in\mathcal{L}}\hat{\mathcal{B}}_{m,n}^{(\hat{i}_t)}$ denotes the binary search space for matrix $\mathbf{B}$ corresponding to the selected node $\hat{\mathcal{N}}_{\hat{i}_t}$ and $\hat{\mathcal{B}}_{m,n}^{(\hat{i}_t)}$ is the search space for binary variable $b_m[n]$. Here, $\hat{\mathcal{B}}^{(\hat{i}_t)}$ contains feasible sets corresponding to both the determined and undetermined binary variables. To simplify notation, we define $\hat{\mathcal{U}}^{(\hat{i}_t)}$ and $\hat{\mathcal{D}}^{(\hat{i}_t)}$ as the collections of the indices of undetermined and determined binary variables in $\hat{\mathcal{B}}^{(\hat{i}_t)}$, respectively.

Then, following a similar procedure as in Section \rom{3}-C, we partition the selected node $\hat{\mathcal{N}}_{\hat{i}_t}$ into two child nodes based on the Euclidean distance between the optimal solution for the lower bound $\hat{\mathbf{B}}_L^{(\hat{i}_t)}$ and its binary version $\hat{\mathbf{B}}_U^{(\hat{i}_t)}$. Specifically, we first obtain the index $(\hat{m}^{(t)},\hat{n}^{(t)})$ based on 
\begin{equation}
   (\hat{m}^{(t)},\hat{n}^{(t)})=\arg\max_{m,n}|\hat{b}^{(\hat{i}_t)}_{L,m}[n]-\hat{b}^{(\hat{i}_t)}_{U,m}[n]|,\vspace*{-1mm}
\end{equation}
where $\hat{b}^{(\hat{i}_t)}_{L,m}[n]$ and $\hat{b}^{(\hat{i}_t)}_{U,m}[n]$ represent the $n$-th element in the $m$-th column of $\hat{\mathbf{B}}_L^{(\hat{i}_t)}$ and $\hat{\mathbf{B}}_U^{(\hat{i}_t)}$, respectively.
Next, the feasible set for $\mathbf{B}$ of the selected node, i.e.,  $\hat{\mathcal{B}}^{(\hat{i}_t)}$ is divided into two new subsets corresponding to $b_{\hat{m}^{(t)}}[\hat{n}^{(t)}]=0$ and $b_{\hat{m}^{(t)}}[\hat{n}^{(t)}]=1$, respectively.
Then, the resource allocation problems $\hat{\mathcal{P}}_i, i\in\{0,1\}$, for the two child nodes are given by
\begin{eqnarray}
\label{LB_Problem_branch_robust}
    \hspace{-0mm}\hat{\mathcal{P}}_i:&&\hspace*{-7mm}\underset{\bm{\Gamma},\mathbf{B}}{\mino}\hspace*{2mm}\widetilde{P}(\hat{\mathbf{W}},\mathbf{B})\\[-2pt]
    &&\hspace*{-12mm}\mbox{s.t.}\hspace*{1mm} \widehat{\overline{\mbox{C1}}},\overline{\mbox{C2}},\mbox{C3},\mbox{C5},\mbox{C7a},\mbox{C7b},\mbox{C7c},\mbox{C7d},\mbox{C9},\notag\\[-0pt]
    &&\hspace*{-13mm}\widehat{\widetilde{\mbox{C4a}}}\hspace*{-0.5mm}:\hspace*{-0.5mm}b_m[n]\in\hat{\mathcal{B}}_{m,n}^{(\hat{i}_t)},\forall (m,n)\hspace*{-1mm}\neq\hspace*{-1mm}(\hat{m}^{(t)},\hat{n}^{(t)}),\,\widehat{\widetilde{\mbox{C4b}}}\hspace*{-0.5mm}:\hspace*{-0.5mm}b_{\hat{m}^{(t)}}[\hat{n}^{(t)}]\hspace*{-1mm}=\hspace*{-1mm}i.\notag
\end{eqnarray}
Note that $\widehat{\widetilde{\mbox{C4a}}}$ is non-convex due to the discreteness of search space $\hat{\mathcal{B}}^{(\hat{i}_t)}$. Here, we relax the undetermined binary variables ${b}_{m}[n]$, $\forall (m,n)\in\hat{\mathcal{U}}^{(\hat{i}_t)}\setminus(\hat{m}^{(t)},\hat{n}^{(t)})$ to continuous bounded variables, i.e., $0\leq{b}_{m}[n]\leq 1$. Then, a lower bound of \eqref{LB_Problem_branch_robust} is obtained by solving the following relaxed problem $\widehat{\widetilde{P}}_i$, $i\in\{0,1\}$
\begin{eqnarray}
\label{LB_Problem_branch_robust_relaxed}
    \hat{\widetilde{\mathcal{P}}}_i:&&\hspace*{-4mm}\underset{\bm{\Gamma},\mathbf{B}}{\mino}\hspace*{2mm}\widetilde{P}(\hat{\mathbf{W}},\mathbf{B})\notag\\[-1pt]
    &&\hspace*{-4mm}\mbox{s.t.}\hspace*{4mm} \widehat{\overline{\mbox{C1}}},\overline{\mbox{C2}},\mbox{C3},\widehat{\widetilde{\mbox{C4b}}},\mbox{C5},\mbox{C7a},\mbox{C7b},\mbox{C7c},\mbox{C7d},\mbox{C9}\notag,\\[-1pt]
    &&\hspace*{4mm}\widehat{\widetilde{\mbox{C4c}}}:b_m[n]\in\hat{{\mathcal{B}}}_{m,n}^{(\hat{i}_t)},\,\forall (m,n)\in\hat{\mathcal{D}}^{(\hat{i}_t)}\\[-1pt]
    &&\hspace*{4mm}\widehat{\widetilde{\mbox{C4d}}}:0\leq b_{m}[n]\leq 1,\,\forall (m,n)\in\hat{\mathcal{U}}^{(\hat{i}_t)}\setminus(\hat{m}^{(t)},\hat{n}^{(t)}),\notag
\end{eqnarray}
which is a convex problem, and the optimal solution can be attained by
CVX. Here, the optimal ${\mathbf{B}}$ and $\hat{\mathbf{W}}$ for 
$\hat{\widetilde{\mathcal{P}}}_0$ and $\hat{\widetilde{\mathcal{P}}}_1$ are denoted as $(\hat{\mathbf{B}}_{L,0}^{(t)}, \hat{\mathbf{W}}_{L,0}^{(t)})$ and $(\hat{\mathbf{B}}_{L,1}^{(t)}, \hat{\mathbf{W}}_{L,1}^{(t)})$, respectively. Therefore, lower bounds corresponding to problem $\hat{{\mathcal{P}}}_0$ and $\hat{{\mathcal{P}}}_1$ are given by $\widetilde{P}(\hat{\mathbf{W}}_{L,0}^{(t)}, \hat{\mathbf{B}}_{L,0}^{(t)})$ and $\widetilde{P}(\hat{\mathbf{W}}_{L,1}^{(t)}, \hat{\mathbf{B}}_{L,1}^{(t)})$, respectively.

Next, similar as in Section \rom{3}-C, we generate a feasible binary matrix $\hat{\mathbf{B}}_{U,i}^{(t)}$ based on $\hat{\mathbf{B}}_{L,i}^{(t)}$ by employing \textbf{Algorithm 1}, where $i\in\{0, 1\}$. Then, an upper bound corresponding to the new node with $b_{\hat{m}^{(t)}}[\hat{n}^{(t)}]=i$ can be obtained by solving the optimization problem in \eqref{UB_Problem_robust} with $\mathbf{B}=\hat{\mathbf{B}}_{U,i}^{(t)}$, $i \in\{0, 1\}$. Note that problem \eqref{UB_Problem_robust} is convex with given $\mathbf{B}=\hat{\mathbf{B}}_{U,i}^{(t)}$, $i \in\{0, 1\}$, and can be optimally solved. We let $\hat{\mathbf{W}}_{U,i}^{(t)}$ denote the corresponding optimal solution of $\hat{\mathbf{W}}$ obtained by solving \eqref{UB_Problem_robust} with $\mathbf{B}=\hat{\mathbf{B}}_{U,i}^{(t)}$, $i\in\{0,1\}$. An upper bound of problem $\hat{{\mathcal{P}}}_i$ in \eqref{LB_Problem_branch_robust} is given by $\widetilde{P}(\hat{\mathbf{W}}_{U,i}^{(t)},\hat{\mathbf{B}}_{U,i}^{(t)})$, $i \in\{0, 1\}$.

{\color{black} The resulting BnB-based algorithm has a similar form as \textbf{Algorithm 2}. In each iteration, we add two nodes and derive the corresponding upper and lower bounds, respectively. In particular, to attain the lower bound for a new node $i$, we solve SDP problem $\hat{\widetilde{\mathcal{P}}}_i$, which involves $K$ LMI constraints of size $(L_k+1)$ due to $\widehat{\overline{\mbox{C1}}}$, $K$ LMI constraints of size $(2N+M)$ due to C7a, $K$ LMI constraints of size $(2M+N)$ due to C7c, and $K$ LMI constraints of size $M$ due to C9. Based on \cite[Theorem 3.12]{bomze2010interior}, the computational complexity of solving one SDP problem is given by $\mathcal{O}(K((L_{\mathrm{max}}+1)^3+(2N+M)^3+(2M+N)^3+M^3)+K^2((L_{\mathrm{max}}+1)^2+(2N+M)^2+(2M+N)^2+M^2)+K^3)$
 , where $L_{\mathrm{max}}=\max_{k\in\mathcal{K}}(L_k)$. On the other hand, to derive the upper bound for a new node $i$, we first obtain the corresponding feasible solutions of $\mathbf{B}$ based on \textbf{Algorithm 1} and then solve SDP problem $\hat{{\mathcal{P}}}_i$ for a given $\mathbf{B}$. Similar to the analysis in Section \rom{3}-C, the computational complexity for finding a feasible solution for each node is given by $\mathcal{O}(I_{\mathrm{itr}}((\frac{M\times (M-1)}{2})(M\times N+1)^3+(\frac{M\times (M-1)}{2})^2(M\times N+1)^2+(M\times N+1)^3))$. Moreover, since problem $\hat{{\mathcal{P}}}_i$ involves the same SDP constraints as problem $\hat{\widetilde{\mathcal{P}}}_i$, the computational complexity for solving $\hat{{\mathcal{P}}}_i$ in each iteration for a given $\mathbf{B}$ is given by
 $\mathcal{O}(K((L_{\mathrm{max}}+1)^3+(2N+M)^3+(2M+N)^3+M^3)+K^2((L_{\mathrm{max}}+1)^2+(2N+M)^2+(2M+N)^2+M^2)+K^3)$ as well. Thus, the overall computational complexity of each new node is given by $\mathcal{O}(2K((L_{\mathrm{max}}+1)^3+(2N+M)^3+(2M+N)^3+M^3)+2K^2((L_{\mathrm{max}}+1)^2+(2N+M)^2+(2M+N)^2+M^2)+2K^3+I_{\mathrm{itr}}((\frac{M\times (M-1)}{2})(M\times N+1)^3+(\frac{M\times (M-1)}{2})^2(M\times N+1)^2+(M\times N+1)^3))$.}
 Moreover, the BnB tree update, the convergence, and the optimality analysis are identical to those in Section \rom{3}-C, and are not described here in detail. 

{\color{black} 
\begin{remark}
     We note that the BnB-based algorithm designed for imperfect CSI can be also applied to solve the resource allocation problem for perfect CSI by setting the CSI estimation error to zero, i.e., $\epsilon_k=0$. However, the resulting computational complexity is significantly higher than that of the algorithm proposed for perfect CSI in Section \rom{3}. The higher computational complexity is due to the additional LMI constraints and auxiliary optimization variables that had to be introduced to appropriately reformulate the robust QoS constraint needed to cope with imperfect CSI.
\end{remark}}
\subsection{Proposed SCA-based Algorithm}
Following similar steps as for derivation of \eqref{SCA_begin} and \eqref{SCA_end}, applying SDP relaxation to address rank-one constraint C8 and exploiting SCA, 
we obtain the following optimization problem to be solved in the $(j+1)$-th iteration of the SCA algorithm
\begin{eqnarray}
\label{Reform_Problem_robust_SCA}
    &&\hspace*{-7mm}\underset{\substack{\bm{\Gamma},\mathbf{B}}}{\mino}\hspace*{1mm}\widetilde{P}(\hat{\mathbf{W}},\mathbf{B})\hspace*{-0.5mm}+\hspace*{-0.5mm}\frac{1}{\mu}\hspace*{-0.5mm}\sum_{m,n}\hspace*{-0.5mm}b_m[n]\hspace*{-0.5mm}-\hspace*{-0.5mm}2b_m^{(j)}[n]b_m[n]\hspace*{-0.5mm}+\hspace*{-0.5mm}(b_m^{(j)}[n])^2\notag\\[-0pt]
    &&\hspace*{2mm}\mbox{s.t.}\hspace*{7mm} \widehat{\overline{\mbox{C1}}},\overline{\mbox{C2}},\mbox{C3},\overline{\mbox{C4}},\mbox{C5},\mbox{C7a},\mbox{C7b},\mbox{C7c},\mbox{C7d},\mbox{C9},
\end{eqnarray}
where $b^{(j)}_m[n]$ denotes the solution of $b_m[n]$ obtained in the $j$-th iteration. The above problem is convex and can be optimally solved by CVX. The proof of the tightness of the SDP relaxation is similar to the proof of Theorem 1. The proposed SCA-based algorithm can be summarized in a similar manner as in \textbf{Algorithm 3}. In particular, optimization problem \eqref{Penalty_problem_bnb} is replaced by \eqref{Reform_Problem_robust_SCA}. Moreover, problem \eqref{Reform_Problem_robust_SCA} is similar to problem $\hat{\widetilde{P}}_i$, $i\in\{0,1\}$ in \eqref{LB_Problem_branch_robust_relaxed}, with a slight difference in the objective function. Therefore, the computational complexity in each iteration for solving one SDP problem is given by $\mathcal{O}(\log\frac{1}{\rho}(K((L_{\mathrm{m}}+1)^3+(2N+M)^3+(2M+N)^3+M^3)+K^2((L_{\mathrm{m}}+1)^2+(2N+M)^2+(2M+N)^2+M^2)+K^3))$\cite[Theorem 3.12]{bomze2010interior}. The initial point and convergence analysis are identical to those in Section \rom{3}-D, and thus are omitted here. 
\section{Numerical Results}
\subsection{Simulation Setting}
In this paper, we consider an MA-enabled multiuser MISO system, where the BS is equipped with $M=4$ MA elements to serve $K=4$ single-antenna users unless otherwise stated. The carrier frequency is set to $5$ GHz, i.e., the wavelength is $\lambda=60$ mm. The transmit area is a square area of size $l\lambda\times l\lambda$, where $l=2$ denotes the normalized transmit area size at the BS. The minimum distance $D_{\mathrm{min}}$ is set to $15$ mm. The distances of the users to the BS are uniformly distributed between $20$ m to $80$ m. The noise variance is set to $-80$ dBm, $\forall k \in\mathcal{K}$. The channel coefficient $h_k(\mathbf{p}_n)$ between an MA element at position $\mathbf{p}_n$ and user $k$ is modeled by the field response channel model detailed in Section \rom{2}-A.  Besides, since we consider isotropic scattering environments, the MPCs are uniformly distributed over the half-space in front of the antenna panel, which yields elevation and azimuth AoDs, $\theta_{k,l_k}$ and $\phi_{k,l_k}$, of the $l_k$-th channel path for user $k$ following probability density function $f_{\mathrm{AoD}}(\theta_{k,l_k}, \phi_{k,l_k})=\frac{\cos\theta_{k,l_k}}{2\pi}$, $\theta_{k,l_k}\in [- \pi/2,\pi/2]$, $\phi_{k,l_k}\in [- \pi/2,\pi/2]$ \cite{zhu2022modeling}. All $L_k=16$ path coefficients $\Psi_{l_k,k}$ are independently and identically distributed and follow complex Gaussian distribution $\mathcal{CN}(0,L_0D_k^{-\alpha})$, where $L_0$, $D_k$, and $\alpha=2.2$ denote the path loss at reference distance $d_0=1$ m, the distance from BS to user $k$, and the path loss exponent, respectively.
In this work, the initial positions of the MA elements are randomly generated. Moreover, the power consumption of one MA driver $P_{\mathrm{MA}}$ is set as $8$ W \cite{AM3248}, and the time durations of two subframes are $T_{\mathrm{MA}}=30$ ms and $T_{\mathrm{Data}}=270$ ms, respectively\footnote{\color{black} The actual power consumption and movement delay of MA depends on the specific hardware configuration, which may vary across different systems. For example, utilizing electronically driven MAs, as proposed in \cite{ning2024movable}, can significantly reduce both power consumption and movement delay at the expense of increased hardware cost.}. 
In addition, high-speed stepper motors combined with linear actuators are adopted as MA drivers in this work\cite{mclean1988review}, and the angular velocity and step angle of the stepper motors are set to $\omega=60\pi\, \mathrm{rad/s}=720\,\mathrm{step/s}$ and $\omega_d=\pi/12$ rad\cite{AM3248}, respectively. A lead screw with an outer diameter of $D=10$ mm is adopted as the linear actuator of the MA drivers. Therefore, the maximum speed and step size of MA motion are given by $v_{\mathrm{v,MA}}=v_{\mathrm{h,MA}}=v_{\mathrm{MA}}=\omega D= 0.94$ $\mathrm{mm}/\mathrm{ms}$ and $\omega_d D\approx 2$ mm, respectively\cite{mclean1988review}.
Thus, the maximal moving distance of an MA is given by $D_{\mathrm{v,max}}=D_{\mathrm{h,max}}=v_{\mathrm{MA}}T_{\mathrm{MA}}\approx28$ mm. Without loss of generality, we assume all users impose the same SINR requirement, i.e., $\gamma_k=\gamma, \forall k$, and define the maximum normalized estimation error of the PCV as $\kappa_k=\epsilon_k/\|\bar{\bm{\psi}}_k\|_F=\kappa$, $\forall k$. The convergence tolerances $\Delta_{\mathrm{BnB}}$ and $\Delta_{\mathrm{SCA}}$ are set as $\Delta_{\mathrm{BnB}}=\Delta_{\mathrm{SCA}}=10^{-4}$, respectively. The number of channel realizations is $200$.

We consider four baseline schemes for comparison. For baseline scheme 1, the MA elements are fixed at $M$ positions and are chosen randomly from $\mathcal{P}$, i.e., binary matrix $\mathbf{B}$ is randomly chosen but satisfies constraints C2, C3, C4, and C5. The beamforming vectors $\mathbf{w}_k$ are obtained by solving the problem 
in \eqref{Reform_Problem} and \eqref{Reform_Problem_robust} for the randomly chosen $\mathbf{B}$ for perfect and imperfect CSI, respectively. For baseline scheme 2, we adopt the AS technique, where the BS is equipped with a $2\times M$ uniform planar array with fixed-position antennas spaced by $\lambda/2=30$ mm. We solve the problem in \eqref{Reform_Problem} and \eqref{Reform_Problem_robust} for all possible subsets of $M$ antenna elements and select the optimal subset that minimizes the BS transmit power, respectively. 
As baseline scheme 3, we adopt a suboptimal AO-based algorithm. Specifically, in each iteration, beamforming matrix $\mathbf{W}$ is optimized by solving the problems in \eqref{Reform_Problem} and \eqref{Reform_Problem_robust} with the fixed binary decision matrix $\mathbf{B}$ obtained in the last iteration for perfect and imperfect CSI, respectively. Afterwards, $\mathbf{B}$ is updated by solving the problems in \eqref{Reform_Problem} and \eqref{Reform_Problem_robust} by relaxing constraint C4 for the fixed $\mathbf{W}$ obtained in the current iteration. After convergence, the obtained continuous binary matrix is quantized to the feasible binary selection matrix $\mathbf{B}_{\mathrm{AO}}$, and the beamforming matrix is designed by solving the problems in \eqref{Reform_Problem} and \eqref{Reform_Problem_robust} with $\mathbf{B}_{\mathrm{AO}}$.  For baseline scheme 4, we jointly design $\mathbf{B}$ and $\mathbf{W}$ without taking the motion power consumption into account. In particular, we optimally solve the problem in \eqref{Ori_Problem} and \eqref{Ori_Problem_robust} with a modified objective function, i.e.,  $P=\sum_{k\in\mathcal{K}}\|\mathbf{w}_k\|^2_2$, using the proposed BnB-based method for perfect and imperfect CSI, respectively. 
{\color{black}For baseline scheme 5, we adopt a fully-connected HBF, where the BS is equipped with $2\times M$ antennas and $M$ RF chains. Then, we employ the AO-based HBF optimization method in \cite{yu2016alternating} for beamformer design. In particular, we first design the optimal fully digital beamforming matrix $\mathbf{W}_{\mathrm{opt}}$ by solving an SDP beamforming problem. Then, we jointly design the analog RF precoder $\mathbf{W}_{\mathrm{RF}}$ and the digital baseband precoder $\mathbf{W}_{\mathrm{BB}}$ to minimize the difference between their product and $\mathbf{W}_{\mathrm{opt}}$, i.e., $\|\mathbf{W}_{\mathrm{opt}}-\mathbf{W}_{\mathrm{RF}}\mathbf{W}_{\mathrm{BB}}\|_F$. Specifically, in each iteration, we update the digital baseband precoder $\mathbf{W}_{\mathrm{BB}}$ by CVX \cite{grant2008cvx} given the fixed analog RF precoder $\mathbf{W}_{\mathrm{RF}}$ obtained in the last iteration, and then optimize $\mathbf{W}_{\mathrm{RF}}$ given the previously obtained $\mathbf{W}_{\mathrm{BB}}$ using Manopt toolbox \cite{boumal2014manopt}. After convergence, to satisfy the minimum required SINR constraint, we update the digital baseband precoder by solving the SDP beamforming problem given the fixed $\mathbf{W}_{\mathrm{RF}}$ obtained in the last iteration.}
\subsection{Numerical Results for Perfect CSI}
\subsubsection{Convergence of the proposed algorithms:}
\begin{figure}\vspace{-0.0cm}
    \centering
    \includegraphics[width=1.9in]{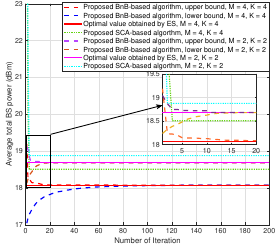}
    \caption{Convergence behavior for different values of $M$ and $K$ and perfect CSI. The system parameters are set as $\gamma=5$ dB, and $d=10$ mm.}\vspace{1mm}
    \label{fig:conver_perfect_CSI}
\end{figure}
Fig. \ref{fig:conver_perfect_CSI} illustrates the convergence of the proposed BnB-based and SCA-based algorithms for different numbers of MA elements $M$ and users $K$ assuming perfect CSI at the BS. Furthermore, we adopt an ES approach, serving as a benchmark, to verify the global optimality of the proposed BnB-based method. Specifically, the ES method solves the problem in \eqref{Reform_Problem} for all ${N}\choose{M}$ possible position combinations of $M$ MA elements to attain the global optimum.  Due to the high complexity of the ES method, we investigate its performance for a system with a relatively large step size $d=10$ mm. As shown in Fig. \ref{fig:conver_perfect_CSI}, the proposed BnB-based method converges to the optimal solution obtained with the ES method, which validates the global optimality of the proposed BnB-based method. Moreover, for the system with $4$ MA elements, the proposed BnB-based algorithm requires only an average of $70$ iterations to converge to the global optimum, which is significantly faster than the ES method that requires solving \eqref{Reform_Problem} ${N}\choose{M}$ $=$ ${13^2}\choose{4}$ times. On the other hand, the proposed SCA-based algorithm converges to a locally optimal value within 10 iterations, revealing an excellent trade-off between computational complexity and performance.
\begin{table*}[!htbp]
\color{black}
\caption{Normalized runtime comparison for different schemes with perfect CSI.}
\centering
\vspace{1mm}
\begin{tabular}{|l|llll|llll|}
\hline
\multicolumn{1}{|c|}{\multirow{2}{*}{Proposed method}} & \multicolumn{4}{l|}{Average normalized runtime}                                                     & \multicolumn{4}{l|}{Average number of iterations}                                                        \\ \cline{2-9} 
\multicolumn{1}{|c|}{}                                 & \multicolumn{1}{l|}{$l=1$} & \multicolumn{1}{l|}{$l=1.5$} & \multicolumn{1}{l|}{$l=2$} & \multicolumn{1}{l|}{$l=2.5$} & \multicolumn{1}{l|}{$l=1$} & \multicolumn{1}{l|}{$l=1.5$} & \multicolumn{1}{l|}{$l=2$} & \multicolumn{1}{l|}{$l=2.5$} \\ \hline
BnB-based method & \multicolumn{1}{l|}{0.43}    & \multicolumn{1}{l|}{1.18} & \multicolumn{1}{l|}{2.50} & \multicolumn{1}{l|}{4.92}  & \multicolumn{1}{l|}{17.9}    & \multicolumn{1}{l|}{37.4}    & \multicolumn{1}{l|}{69.2}     &   \multicolumn{1}{l|}{92.2}   \\ \hline
SCA-based method & \multicolumn{1}{l|}{0.06}    & \multicolumn{1}{l|}{0.14} & \multicolumn{1}{l|}{0.28} & \multicolumn{1}{l|}{0.47}    & \multicolumn{1}{l|}{2.7}    & \multicolumn{1}{l|}{3.0}    & \multicolumn{1}{l|}{3.5}     &   \multicolumn{1}{l|}{4.7}   \\ \hline
Baseline scheme 3  & \multicolumn{1}{l|}{0.15}    & \multicolumn{1}{l|}{0.46} & \multicolumn{1}{l|}{1} & \multicolumn{1}{l|}{2.01}     & \multicolumn{1}{l|}{6.2}    & \multicolumn{1}{l|}{13.7}    & \multicolumn{1}{l|}{21.2}     &  \multicolumn{1}{l|}{32.5}    \\ \hline
\end{tabular}\vspace{-5mm}
\color{black}
\end{table*}

{\color{black}Next, the normalized runtimes of the proposed algorithms, and the AO-based baseline scheme 3 are investigated for different normalized transmit area sizes $l$ for the case of perfect CSI. In particular, the convergence tolerances of the proposed GBD-based and SCA-based algorithms are set to $\Delta=10^{-2}$ and $\Delta_{\mathrm{SCA}}=10^{-2}$, respectively. On the other hand, the convergence criterion for baseline scheme 3 is given by $\frac{\|\bm{W}_{\mathrm{AO}}^{(i)}-\bm{W}_{\mathrm{AO}}^{(i-1)}\|}{\|\bm{W}_{\mathrm{AO}}^{(i-1)}\|_F}\leq 10^{-2}$. Here, we adopt the runtime of baseline scheme 3 with $l=2$ as the reference time for normalization. The algorithms are implemented in MATLAB R2018a and tested on a PC with a Core i9-13900K CPU. MOSEK was adopted as the CVX solver for time efficiency. As can be observed from Table I, the proposed SCA-based algorithm requires significantly less runtime compared to the BnB-based algorithm and AO-based baseline scheme. This confirms the computational time efficiency of the proposed SCA-based algorithm. This property is attributed to the fact that the proposed SCA-based algorithm needs only a few iterations to converge.}
\subsubsection{Average BS power versus minimum required SINR:} 
\begin{figure}
    \centering
    \includegraphics[width=0.85\linewidth]{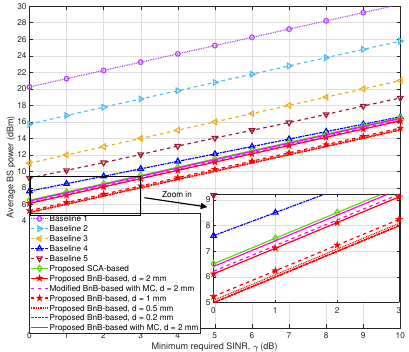}
    \caption{Average BS power consumption versus the minimum required SINR for perfect CSI. The parameters are set as $M=4$, $K=4$, and $l=2$.}\vspace{-2mm}
    \label{fig:SINR_perfectCSI}
\end{figure}
\begin{figure}
    \centering
         \includegraphics[width=2.4in]{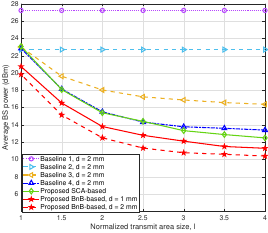}
    \caption{Average BS power consumption versus the normalized transmit area size of the MA-enabled system for perfect CSI. The parameters are set as $M=4$, $N=4$, and $\gamma=7$ dB.}\vspace{-2mm}
    \label{fig:Region_perfectCSI}
\end{figure}
Fig. \ref{fig:SINR_perfectCSI} shows the average BS power consumption versus the minimum required SINR $\gamma$ for perfect CSI. As can be observed, as the minimum required SINR increases, the BS must consume more power to satisfy the more stringent QoS requirements of the users. Moreover, both proposed schemes outperform all five baseline schemes for the entire range of $\gamma$. In particular, since for baseline schemes 1 and 2, the BS is equipped with an antenna array with fixed antenna positions, the spatial correlation of the transmit antenna array cannot be shaped in an optimal manner. 
{\color{black} Moreover, although the HBF-based baseline scheme 5 outperforms the AS-based baseline scheme 2 due to the extra DoFs provided by its fully-connected structure, the proposed BnB-based scheme achieves a $3$ dB power gain compared to baseline scheme 5. Additionally, the MA-enabled system requires fewer antenna elements and no phase shifters, while achieving a better performance, highlighting the efficiency of MAs compared to HBF.}
Furthermore, for baseline scheme 3, the adopted AO algorithm tends to search only a subset of the feasible region close to the initial point and is likely to get stuck in a spurious stationary point. Thus, the proposed BnB-based scheme achieves a considerable gain ($4$ dB) compared to the AO-based scheme. On the other hand, baseline scheme 4 does not account for the power consumption of the MA drivers, resulting in a $1-2$ dB performance gap compared to the optimal BnB-based scheme. In particular, the gap between baseline scheme 4 and the BnB-based scheme is reduced as the minimum required SINR increases. In fact, as the minimum required SINR increases, the transmit power consumption becomes more dominant in the overall power consumption and eventually outweighs the impact of the power consumed by the antenna movement. Therefore, baseline scheme 4 yields a solution closer to the global optimum one for very high required SINRs. Furthermore, the proposed SCA-based scheme achieves a performance close to that of the optimal BnB-based scheme. 
{\color{black}Moreover, to investigate the maximum performance of the considered discrete MA-enabled system, we further reduce the step size of the MA driver. In particular, we observe from Fig. 5 that reducing the step size of the MA drivers from $2$ mm to $1$ mm leads to a $1.1$ dB performance improvement. This is attributed to the enlarged feasible set for MA positioning due to the refinement of the step size. Then, further reducing the step size to $0.2$ mm produces only negligible additional gains, suggesting that discrete MA-enabled systems with $0.2$ mm step size can closely approach the performance of an ideal system with continuous MA movement.  These observations suggest that a 2 mm step size effectively balances BS power consumption and MA control precision in MA-enabled systems, hence this step size is adopted in the subsequent simulations. }

{\color{black} 
Moreover, we also investigate the impact of the mutual coupling (MC) between MA elements. In particular, according to the empirical model in \cite{chen2018Mutualcoupling,savy2016couplingeffect}, the MC coefficient $C_{i,i'}$ between MA elements $i$ and $i'$ depends on their distance $d_{i,i'}$, and is given by $C_{i,i'}(d_{i,i'})=e^{-\frac{2d_{i,i'}}{\lambda}(\alpha_{\mathrm{mc}}+j\pi)}$,
where $\alpha_{\mathrm{mc}}$ is a hardware parameter controlling the level of coupling. Next, we define MC matrix $\hat{\mathbf{C}}\in\mathbb{C}^{M\times M}$ to collect all MC coefficients $C_{i,i'}$ between all MA elements. Then, the received signal $y_k$ of user $k$ can be rewritten as $y_k=\hat{\mathbf{h}}_{k}^H\mathbf{B}\hat{\mathbf{C}}\mathbf{W}\mathbf{s}+n_k$.
To evaluate the impact of MC, we investigate the optimal performance of the MA-enabled system under MC. More details on optimal resource allocation design under MC are provided in the Appendix. Note that this optimal resource allocation design provides a performance upper bound for MA-enabled systems with MC. Moreover, we also investigate the performance of the proposed algorithm under MC.
In this case, the binary selection matrix $\mathbf{B}$ is designed based on the ideal assumption of no MC by the proposed BnB-based \textbf{Algorithm 2}. Then, to satisfy the minimum SINR requirement with MC, we design the beamforming matrix $\mathbf{W}$ considering the MC effect by optimally solving \eqref{Ori_Problem} for the attained $\mathbf{B}$ and a modified SINR, i.e., $\mathrm{SINR}_k=\frac{|\hat{\mathbf{h}}_{k}^H\mathbf{B}\hat{\mathbf{C}}\mathbf{w}_k|^2}{\sum_{k'\in\mathcal{K}\setminus{k}}|\hat{\mathbf{h}}_{k}^H\mathbf{B}\hat{\mathbf{C}}\mathbf{w}_{k'}|^2+\sigma_k^2}$. 

For our simulations, we adopt a typical value for the MC coefficient $|C_{i,i'}(d_{i,i'}=\lambda)|=0.2$ \cite{liao2012adaptive}. The corresponding value of $\alpha_{\mathrm{mc}}$ is given by $\alpha_{\mathrm{mc}}=-\ln(|C_{i,i'}(d_{i,i'}=\lambda)|)/2\approx0.75$.  
For MC, the average BS power consumption for optimal resource allocation (as outlined in the Appendix) and the proposed BnB-based design (which does not account for MC) is illustrated in Fig. 5. As can be observed, the proposed BnB-based scheme requires a negligible $0.3$ dB additional BS power to compensate the MC effect, compared with the proposed BnB-based scheme without MC. In fact, even without considering the MC effect, the proposed BnB-based algorithm typically positions the MA elements at a comparatively large distance to avoid generation of a rank-deficient channel matrix, leading to low MC. On the other hand, if MC is present, the proposed BnB-based design (which ignores MC) achieves a performance approaching the optimal performance bound achieved by the optimal resource allocation design accounting for MC. However, the optimal resource allocation design accounting for MC introduces a series of additional auxiliary variables, increasing computational complexity. Since MC has a negligible impact on the performance of the MA-enabled system, we ignore MC for algorithm design and in the remaining simulation results. }

\subsubsection{Average BS power versus transmit area size:}
Fig. \ref{fig:Region_perfectCSI} depicts the average BS power consumption for different normalized transmit area sizes for perfect CSI. It is observed that the BS power consumption for baseline schemes 3 and 4 and the proposed schemes decreases as the normalized transmit area size increases. This can be attributed to the fact that a larger transmit area provides more flexibility for positioning the MA elements to adaptively shape the desired spatial correlation, leading to a performance gain. On the other hand, baseline schemes 1 and 2, with their fixed antenna positions, cannot benefit from a larger transmit area, indicating that baseline schemes 1 and 2 cannot fully utilize the spatial DoFs provided by larger transmit areas. Moreover, as can be observed, both the proposed SCA-based and BnB-based schemes outperform all four baseline schemes. In particular, for larger transmit areas, the gap between baseline scheme 4 and the proposed BnB-based optimal design is enlarged to $1.5$ dB. More specifically, increasing the transmit area size leads to a higher power consumption for MA motion. Thus, the impact of the MA motion power on the total power consumption grows and the performance gain of the proposed BnB-based scheme compared to baseline scheme 4 increases.
Furthermore, the performance of the proposed schemes saturates when the normalized transmit area size reaches $l=3$ due to the speed limitation of the MA drivers, which limits the adaptability within one frame. 
\subsubsection{Average BS power and energy efficiency versus the antenna movement time:}
\begin{figure*}
    \centering
    \begin{minipage}{0.31\linewidth}
        \centering
       \includegraphics[width=2.0in]{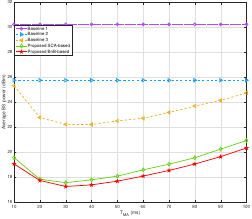}
    \caption{Average BS power consumption versus antenna movement time $T_{\mathrm{MA}}$ for perfect CSI. The parameters are set as $\gamma=10$ dB, $M=4$, $K=4$, and $l=2$.}\vspace{-3mm}
    \label{fig:TMA_perfectCSI}
    \end{minipage}\hspace{2mm}
        \begin{minipage}{0.31\linewidth}
               \includegraphics[width=2.1in]{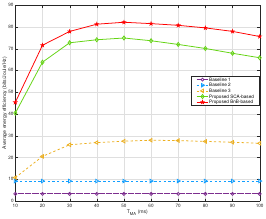}
    \caption{Average energy efficiency versus antenna movement time $T_{\mathrm{MA}}$ for perfect CSI. The parameters are set as $\gamma=10$ dB, $M=4$, $K=4$, and $l=2$.}\vspace{-3mm}
    \label{fig:EE_perfectCSI}
    \end{minipage}\hspace{2mm}
            \begin{minipage}{0.32\linewidth}\vspace{1.5mm}
\centering
         \includegraphics[width=2.0in]{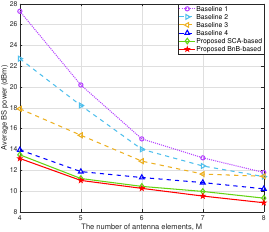}
    \caption{Average BS power consumption versus the number of MA elements of the MA-enabled system for perfect CSI. The parameters are set as $K=4$, $l=2$, $d=2$ mm, and $\gamma=7$ dB.}\vspace{-3mm}
    \label{fig:M_perfectCSI}
    \end{minipage}
    \label{fig:enter-label}
\end{figure*}
{\color{black} Fig. 7 illustrates the BS power consumption as a function of the required antenna movement time $T_{\mathrm{MA}}$ for perfect CSI. Since the MA-enabled system allocates $T_{\mathrm{MA}}$ ms for MA movement at the beginning of each transmission block, the achievable rate of the MA-enabled systems is given by $\frac{T_{\mathbf{Data}}}{T_{\mathbf{Data}}+T_{\mathbf{MA}}}\sum_{k\in\mathcal{K}}\log(1+\gamma)$, which is lower than that of baseline schemes 1 and 2. To fairly evaluate the power efficiency of MA-enabled systems, we set the minimum required SINR of all users in the proposed schemes and baseline scheme 3 as $(\gamma+1)^{\frac{T_{\mathbf{Data}}+T_{\mathbf{MA}}}{T_{\mathbf{Data}}}}-1$ to compensate the rate loss. As can be observed, despite the rate loss, MA-enabled systems still achieve a significant performance gain compared to baseline schemes 1 and 2 with fixed position antennas. Moreover, for small antenna movement times, e.g., $T_{\mathrm{MA}}\leq 30$ ms, increasing $T_{\mathrm{MA}}$ enhances the performance of the proposed BnB-based and SCA-based schemes as well as that of baseline scheme 3. Specifically, a larger antenna movement time expands the feasible region of the positions of the MA elements, rendering the proposed schemes and baseline scheme 3 more likely to find better positions that further reduce BS power consumption. Moreover, for $T_{\mathrm{MA}}>30$ ms, the BS power consumption of the proposed schemes and baseline scheme 3 is enlarged by increasing $T_{\mathrm{MA}}$. This is attributed to the fact that increasing the antenna movement time reduces the time available for data transmission, inevitably leading to system data rate degradation. Thus, the MA-enabled systems have to invest more power to compensate for this rate degradation. This observation highlights a non-trivial trade-off between spectral efficiency and energy efficiency, which has to be carefully considered in practical applications.}

{\color{black}In Fig. 8, we investigate the energy efficiency of different schemes for various antenna movement times under perfect CSI. To ensure a fair comparison, we define the energy efficiency of the fixed-position antenna and the MA-enabled systems as $\frac{\sum_{k\in\mathcal{K}}\log_2(1+\gamma_k)}{\sum_{k\in\mathcal{K}}\|\mathbf{w}_k\|_2^2}$  bits/Joule/Hz and $\frac{T_{\mathrm{Data}}\sum_{k\in\mathcal{K}}\log_2(1+\gamma_k)}{\bar{P}(\mathbf{W},\mathbf{B})(T_{\mathrm{MA}}+T_{\mathrm{Data}})}$ bits/Joule/Hz, respectively, where $\gamma_k$ denotes the minimum required SINR of user $k$. As can be observed, the proposed BnB-based and SCA-based schemes, along with baseline scheme 3, achieve significantly higher energy efficiencies compared to baseline schemes 1 and 2. Additionally, increasing the antenna movement time from $10$ ms to $40$ ms improves the energy efficiency of the proposed schemes and baseline scheme 3, albeit at the expense of a reduced data rate. However, for $T_{\mathrm{MA}}>40$ ms, further increasing $T_{\mathrm{MA}}$ does not yield any additional energy efficiency gains. Specifically, for $T_{\mathrm{MA}}=40$ ms, the feasible region of the positions of the MA elements covers the entire transmit area, and cannot be further enlarged by increasing $T_{\mathrm{MA}}$. Thus, by increasing $T_{\mathrm{MA}}$ further, no reduction of the BS power consumption can be realized while the achievable data rate is decreasing, leading to a lower achieved energy efficiency.
This observation suggests that a moderate antenna movement time, e.g., $20$ ms $<T_{\mathrm{MA}}<40$ ms, is beneficial in practical applications.}
\subsubsection{Average BS power versus the number of MA elements:}
We investigate the BS power consumption as a function of the number of MA elements $M$ in Fig. \ref{fig:M_perfectCSI}. As can be observed, the system performance of all schemes improves significantly by adding more antenna elements. This can be attributed to the extra diversity gain provided by the additional antenna elements. Moreover, the performance gap between the proposed schemes and baseline schemes 1, 2, and 3 is enlarged significantly as the number of antenna elements decreases.  Specifically, in a high-load system with $K=M$, baseline schemes 1 and 2 struggle to efficiently suppress the inter-user interference due to the strong correlation of the users' channel vectors.  On the other hand, as MAs can reshape the channel correlation, both the proposed SCA- and BnB-based schemes can reduce inter-user interference by optimizing the MA positions. Moreover, compared with AO-based baseline scheme 3, the proposed SCA- and BnB-based methods achieve a $5$ dB performance gain in a high-load system with $K=M$, which highlights the importance of an efficient joint optimization of the MA position and BS beamforming in such a scenario.
\subsection{Numerical Results for Imperfect CSI}
\begin{figure*}
    \centering
    \begin{minipage}{0.31\linewidth}
        \centering
        \includegraphics[width=1.9in]{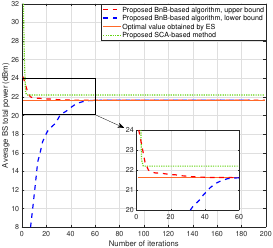}\vspace{-0mm}
    \caption{Convergence behavior of the proposed algorithms for imperfect CSI. The system parameters are set as $M=2$, $K=2$, $\gamma=5$ dB, $d=10$ mm, and $\kappa = 0.1$.}\vspace{-6mm}
    \label{fig:conver_imperfect_CSI}
    \end{minipage}\hspace{2mm}
        \begin{minipage}{0.31\linewidth}
        \centering
         \includegraphics[width=2.0in]{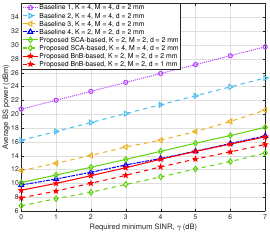}
    \caption{Average BS power consumption versus the minimum required SINR of the users for imperfect CSI. The system parameters are set as $l=2$ and $\kappa = 0.1$.}\vspace{-6mm}
    \label{fig:SINR_imperfectCSI}
    \end{minipage}\hspace{2mm}
            \begin{minipage}{0.32\linewidth}\vspace{1.5mm}
        \centering
         \includegraphics[width=1.99in]{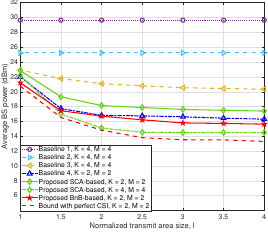}
    \caption{Average BS power consumption versus the normalized transmit area size of the MA-enabled system for imperfect CSI. The system parameters are set as $\gamma = 7$ dB, $d=2$ mm, and $\kappa=0.1$.}\vspace{-6mm}
    \label{fig:Region_imperfectCSI}
    \end{minipage}
    \label{fig:enter-label}
\end{figure*}
\subsubsection{Convergence of the proposed algorithms:}
Fig. \ref{fig:conver_imperfect_CSI} illustrates the convergence behavior of the proposed schemes for imperfect CSI. Here, to verify optimality, we adopt again an ES method as a benchmark. Due to the high computational complexity of the ES algorithm, we only investigate its performance for a MIMO system with $M=2$ and $K=2$. As can be observed, the proposed BnB-based scheme converges to a solution identical to the global optimum obtained by the ES method, which confirms the optimality of the proposed BnB-based method. Moreover, the proposed BnB-based method converges to the global optimum in $60$ iterations, which is much faster than the ES method. On the other hand, the BnB-based algorithm for imperfect CSI requires more iterations to converge compared to the algorithm for perfect CSI, cf. e.g., Fig. \ref{fig:conver_perfect_CSI}. In fact, for imperfect CSI, we had to introduce a series of auxiliary variables and SDP constraints to tackle the robust constraint and the intricate coupling between $\mathbf{W}_k$ and $\mathbf{B}$, which inevitably reduces the speed of convergence. Furthermore, the proposed SCA-based scheme achieves close-to-optimal performance within only $10$ iterations, revealing its computational efficiency for robust system design.

\begin{table*}[!htbp]
\color{black}
\caption{Normalized runtime comparison for different schemes with imperfect CSI.}
\centering
\vspace{1mm}
\begin{tabular}{|l|llll|llll|}
\hline
\multicolumn{1}{|c|}{\multirow{2}{*}{Proposed method}} & \multicolumn{4}{l|}{Average normalized runtime}                                                     & \multicolumn{4}{l|}{Average number of iterations}                                                        \\ \cline{2-9} 
\multicolumn{1}{|c|}{}                                 & \multicolumn{1}{l|}{$l=1$} & \multicolumn{1}{l|}{$l=1.5$} & \multicolumn{1}{l|}{$l=2$} & \multicolumn{1}{l|}{$l=2.5$} & \multicolumn{1}{l|}{$l=1$} & \multicolumn{1}{l|}{$l=1.5$} & \multicolumn{1}{l|}{$l=2$} & \multicolumn{1}{l|}{$l=2.5$} \\ \hline
BnB-based method & \multicolumn{1}{l|}{2.04}    & \multicolumn{1}{l|}{6.03} & \multicolumn{1}{l|}{12.70} & \multicolumn{1}{l|}{20.31}  & \multicolumn{1}{l|}{42.1}    & \multicolumn{1}{l|}{98.4}    & \multicolumn{1}{l|}{229.7}     &   \multicolumn{1}{l|}{402.3}   \\ \hline
SCA-based method & \multicolumn{1}{l|}{0.13}    & \multicolumn{1}{l|}{0.28} & \multicolumn{1}{l|}{0.61} & \multicolumn{1}{l|}{0.84}    & \multicolumn{1}{l|}{2.9}    & \multicolumn{1}{l|}{3.5}    & \multicolumn{1}{l|}{4.1}     &   \multicolumn{1}{l|}{5.2}   \\ \hline
Baseline scheme 3  & \multicolumn{1}{l|}{0.14}    & \multicolumn{1}{l|}{0.36} & \multicolumn{1}{l|}{0.85} & \multicolumn{1}{l|}{2.42}     & \multicolumn{1}{l|}{6.5}    & \multicolumn{1}{l|}{9.7}    & \multicolumn{1}{l|}{15.7}     &  \multicolumn{1}{l|}{24.3}    \\ \hline
\end{tabular}\vspace{-5mm}
\color{black}
\end{table*}

{\color{black}Table \rom{2} illustrates the normalized runtime of the proposed algorithms and baseline scheme 3 w.r.t. the normalized transmit area size $l$. In particular, we adopt the same convergence criteria for the proposed algorithms and baseline scheme 3 as for the case of perfect CSI in Table \rom{1}, and the reference time for normalization is the runtime of baseline scheme 3 with $l=2$ for perfect CSI. Compared with the results in Table \rom{1}, both the SCA-based and BnB-based algorithms for imperfect CSI require longer runtimes to converge even in the 2-user case. In fact, the proposed robust designs introduce additional LMI constraints and auxiliary variables for the reformulation of the robust QoS constraint, leading to a relatively higher computational complexity compared to the algorithms for perfect CSI. Nevertheless, the proposed SCA-based algorithm requires a shorter runtime than the AO-based baseline scheme 3 since the SCA-based method requires fewer iterations to converge.}
\subsubsection{Average BS power versus minimum required SINR:} 
Fig. \ref{fig:SINR_imperfectCSI} shows the average BS power consumption versus the minimum required SINR for imperfect CSI. Due to the exponential growth in the computational complexity of the proposed BnB-based algorithm w.r.t. the number of MA elements $M$, we simulate the performance of the proposed BnB-based algorithm and baseline scheme 4 for a small value of $M$, i.e., $M=2$. As can be observed, the proposed SCA- and BnB-based algorithms reduce the BS power consumption significantly compared to baseline schemes 1, 2, and 3. In particular, for the system with $M=4$ and $K=4$, the proposed SCA-based method achieves a gain of 4-5.5 dB compared to baseline scheme 3. In fact, the results obtained by the AO-based scheme are quantized to a feasible binary selection matrix after convergence. However, coarse quantization leads to mismatches in robust beamforming and interference suppression, causing significant performance degradation. 
On the other hand, the proposed low-complexity SCA-based scheme achieves a performance close to that of the optimal BnB-based scheme for a system with $M=2$ MA elements and $K=2$ users, thereby highlighting the excellent performance-complexity tradeoff of the SCA-based method.

\subsubsection{Average BS power versus the transmit area size:}
In Fig. \ref{fig:Region_imperfectCSI}, we investigate the average BS power consumption as a function of the normalized transmit area size. 
As can be observed, the performance gap between the proposed SCA- and BnB-based schemes widens as the transmit area size increases. 
In fact, increasing the size of the transmit area introduces more binary decision variables that are required to be relaxed in the SCA-based method, which makes the SCA-based method more likely to converge to a local optimum.
Thus, the proposed SCA-based scheme becomes more prone to converging to a local optimum with suboptimal performance. 
Moreover, as can be observed in Fig. \ref{fig:Region_imperfectCSI} the gap between the optimal performances for perfect and imperfect CSI also increases as the size of the transmit region expands. In particular, a larger transmit area size increases the size of the effective channel vectors $\hat{\mathbf{h}}_k$, leading to an enlarged feasible region for the CSI uncertainty. Therefore, the BS has to allocate more transmit power to compensate the larger CSI estimation errors as the size of the transmit region increases.

\subsubsection{CSI uncertainty:}
\begin{figure}\vspace{-0.0cm}
    \centering
    \includegraphics[width=2.4in]{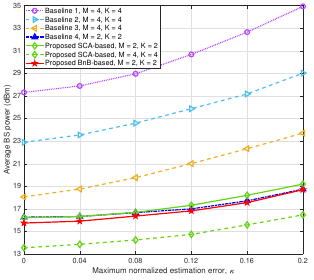}
    \caption{Average BS power versus the maximum channel estimation error. The system parameters are set as $\gamma = 7$ dB and $l=2$.}\vspace{-2mm}
    \label{fig:varcsi_imperfectCSI}
\end{figure}
{\color{black} Fig. \ref{fig:varcsi_imperfectCSI} shows the average BS power consumption versus the maximum normalized channel estimation error $\kappa$. It is observed that as the CSI estimation quality deteriorates, the average BS power consumption increases. In particular, as the accuracy of the estimated CSI decreases, it becomes more challenging for the BS to perform accurate beamforming while satisfying the QoS criteria. As such, a higher transmit power is required to guarantee the robustness of the communication system under imperfect CSI. 
Moreover, as $\kappa$ increases, the gain of the proposed SCA-based scheme compared to baseline schemes 1, 2, and 3 grows, highlighting the superiority of the proposed scheme under severe CSI uncertainty. In addition, the performance achieved by the low-complexity SCA-based scheme is close to the optimal performance achieved by the proposed BnB-based scheme in the $2$-user system. This confirms the efficiency of the proposed SCA-based scheme in minimizing the transmit power in the presence of high CSI uncertainty. }
\section{Conclusion}
In this paper, we investigated for the first time the globally optimal resource allocation design for a multiuser MA-enabled downlink MISO communication system for both perfect and imperfect CSI. Due to practical hardware limitations, we modeled the motion and energy consumption of the MA in a discrete manner. Then, for both perfect and imperfect CSI, we formulated resource allocation problems for minimization of the total BS power consumption, comprising the radiated and MA motion powers, while guaranteeing a minimum required SINR for each user. To tackle the coupling between the MA positions and the beamforming vectors, we leveraged several transformations and auxiliary variables to recast the formulated optimization problems as tractable MINLP problems. To fully reveal the potential of MA-enabled systems, we proposed BnB-based algorithms for the jointly optimal design of the MA positions and BS beamforming for perfect and imperfect CSI, respectively. Moreover, to reduce computational complexity, we leveraged the SCA technique to obtain locally optimal solutions with polynomial time complexity.
Our simulation results verified the global optimality of the proposed BnB-based algorithms and demonstrated the effectiveness of SCA-based algorithms. Moreover, we highlighted the robustness of the proposed BnB-based and SCA-based algorithms to imperfect CSI, ensuring high performance even in the presence of uncertain channel conditions. Furthermore, for both perfect and imperfect CSI, the proposed SCA-based algorithms attained high performance, offering significant advantages for real-time implementation. Additionally, we showed that the proposed MA-enabled systems reduce the BS power consumption significantly compared to conventional fixed-antenna systems, especially in high-load scenarios with a limited number of antenna elements.
\appendix[Optimal Resource Allocation Design Under Mutual Coupling]
{\color{black}To facilitate the optimal resource allocation design under MC, we define matrix $\mathbf{C}=e^{-\frac{2\mathbf{D}}{\lambda}(\alpha_{\mathrm{mc}}+j\pi)}$, where $\mathbf{D}$ is the distance matrix defined in Section \rom{3}. Here, element $D_{n,n'}$ in the $n$-th row and $n'$-th column of $\mathbf{D}$ denotes the distance between the $n$-th and the $n'$-th candidate positions in $\mathcal{P}$. Then, MC coefficient $C_{i,i'}$ can be recast as $\mathbf{b}_i^T\mathbf{C}\mathbf{b}_{i'}$, where $\mathbf{b}_i$ denotes the binary selection vector of MA element $i$. Then, we define MC matrix $\hat{\mathbf{C}}\in\mathbb{C}^{M\times M}$, which collects all MC coefficients $C_{i,i'}$ between all MA elements, i.e., $\hat{\mathbf{C}}=\mathbf{B}^T\mathbf{C}\mathbf{B}$. Taking into account the MC effect, the received signal $y_k$ of user $k$ can be rewritten as 
\begin{equation}
    y_k=\hat{\mathbf{h}}_{k}^H\mathbf{B}\hat{\mathbf{C}}\mathbf{W}\mathbf{s}+n_k=\hat{\mathbf{h}}_{k}^H\mathbf{B}\mathbf{B}^T\mathbf{C}\mathbf{B}\mathbf{W}\mathbf{s}+n_k.
\end{equation}
Then, we introduce a new auxiliary variable $\mathbf{X}_{\mathrm{MC}}\in\mathbb{C}^{N\times K}$, which is given by
\begin{equation}\label{EQ_MC}
\hspace{-4mm}\mathbf{X}_{\mathrm{MC}}=\mathbf{B}\mathbf{B}^T\mathbf{C}\mathbf{X}=\sum_{m}\mathbf{b}_m\mathbf{b}_m^T\mathbf{C}\mathbf{X}\overset{(a)}{=}\mathrm{diag}(\sum_{m}\mathbf{b}_m)\mathbf{C}\mathbf{X},
\end{equation}
where equality (a) is attributed to the properties of binary selection vector $\mathbf{b}_m$. Then, with MC, the received signal of user $k$ can be recast as 
\begin{equation}
y_k=\hat{\mathbf{h}}_{k}^H\mathbf{X}_{\mathrm{MC}}\mathbf{s}+n_k.
\end{equation}
To tackle the nonconvexity of the equality constraint in \eqref{EQ_MC}, we exploit Lemma 1 and reformulate \eqref{EQ_MC} in terms of the following LMI constraints:
\begin{eqnarray}
\mathrm{C6c}\mathrm{:}\hspace*{1mm}
   \begin{bmatrix}
        \mathbf{U}_{\mathrm{MC}} & \hspace{-2mm}\mathbf{X}_{\mathrm{MC}} & \mathrm{diag}\{\sum_{m\in\mathcal{M}}\mathbf{b}_m\}\\
        \mathbf{X}_{\mathrm{MC}}^H & \hspace{-2mm}\mathbf{V}_{\mathrm{MC}} & \mathbf{X}^H\mathbf{C}^H\\
        \mathrm{diag}\{\sum_{m\in\mathcal{M}}\mathbf{b}_m\}^H &\hspace{-2mm} \mathbf{C}\mathbf{X} & \mathbf{I}_M
    \end{bmatrix}&\succeq \mathbf{0},\notag\\
\mathrm{C6d}\mathrm{:}\hspace*{1mm}\label{DC}\mathrm{Tr}\big(\mathbf{U}_{\mathrm{MC}}\big)-M\leq 0,\vspace*{-2mm}\notag
\end{eqnarray}
where $\mathbf{U}_{\mathrm{MC}}\in\mathbb{C}^{N\times N}$ and $\mathbf{V}_{\mathrm{MC}}\in\mathbb{C}^{K\times K}$ are two auxiliary optimization variables with $\mathbf{U}_{\mathrm{MC}}\succeq \mathbf{0}$ and $\mathbf{V}_{\mathrm{MC}}\succeq \mathbf{0}$. Thus, we formulate the resource allocation problem taking into account MC under perfect CSI conditions as follows
\begin{eqnarray}
\label{Reform_Problem_MC_paper}
    &&\hspace*{-4mm}\underset{\mathbf{X},\mathbf{X}_{\mathrm{MC}},\mathbf{W},\mathbf{B},\mathbf{U},\mathbf{V},\mathbf{U}_{\mathrm{MC}},\mathbf{V}_{\mathrm{MC}}}{\mino}\hspace*{2mm}\widebar{P}(\mathbf{W},\mathbf{B})\\[-3pt]
    &&\hspace*{0mm}\mbox{s.t.}\hspace*{12mm} \mbox{C1a},\mbox{C1b}, \overline{\mbox{C2}},{\mbox{C3}},\mbox{C4},\mbox{C5},\mbox{C6a},\mbox{C6b},\mbox{C6c},\mbox{C6d}\notag,
\end{eqnarray}
where $\mathbf{x}_k$ in constraints C1a and C1b is replaced by the $k$-th column of $\mathbf{X}_{\mathrm{MC}}$, i.e., $\mathbf{x}_{\mathrm{MC},k}$. Note that the resource allocation problem in \eqref{Reform_Problem_MC_paper} is an MINLP problem, which is convex for fixed $\mathbf{B}$ and thus can be optimally solved by the proposed BnB-based \textbf{Algorithm 2}.
}
\bibliographystyle{IEEEtran}
\bibliography{reference}

\begin{thebibliography}{10}
\providecommand{\url}[1]{#1}
\csname url@samestyle\endcsname
\providecommand{\newblock}{\relax}
\providecommand{\bibinfo}[2]{#2}
\providecommand{\BIBentrySTDinterwordspacing}{\spaceskip=0pt\relax}
\providecommand{\BIBentryALTinterwordstretchfactor}{4}
\providecommand{\BIBentryALTinterwordspacing}{\spaceskip=\fontdimen2\font plus
\BIBentryALTinterwordstretchfactor\fontdimen3\font minus
  \fontdimen4\font\relax}
\providecommand{\BIBforeignlanguage}[2]{{%
\expandafter\ifx\csname l@#1\endcsname\relax
\typeout{** WARNING: IEEEtran.bst: No hyphenation pattern has been}%
\typeout{** loaded for the language `#1'. Using the pattern for}%
\typeout{** the default language instead.}%
\else
\language=\csname l@#1\endcsname
\fi
#2}}
\providecommand{\BIBdecl}{\relax}
\BIBdecl

\bibitem{mietzner2009multiple}
J.~Mietzner \emph{et~al.}, ``Multiple-antenna techniques for wireless
  communications - a comprehensive literature survey,'' \emph{IEEE Commun.
  Surv. Tuts.}, vol.~11, no.~2, pp. 87--105, 2009.

\bibitem{tsai2014power}
S.-H. Tsai and H.~V. Poor, ``Power allocation for artificial-noise secure
  {MIMO} precoding systems,'' \emph{IEEE Trans. Signal Process.}, vol.~62,
  no.~13, pp. 3479--3493, Jul. 2014.

\bibitem{xu2022robust}
D.~Xu \emph{et~al.}, ``Robust and secure resource allocation for {ISAC}
  systems: A novel optimization framework for variable-length snapshots,''
  \emph{IEEE Trans. Commun.}, vol.~70, no.~12, pp. 8196--8214, Dec. 2022.

\bibitem{huang2020holographic}
C.~Huang \emph{et~al.}, ``Holographic {MIMO} surfaces for {6G} wireless
  networks: Opportunities, challenges, and trends,'' \emph{IEEE Wirel.
  Commun.}, vol.~27, no.~5, pp. 118--125, Oct. 2020.

\bibitem{zhu2022modeling}
L.~Zhu \emph{et~al.}, ``Modeling and performance analysis for movable antenna
  enabled wireless communications,'' \emph{IEEE Trans. Wireless Commun.},
  vol.~23, no.~6, pp. 6234--6250, 2024.

\bibitem{ma2022mimo}
W.~Ma \emph{et~al.}, ``{MIMO} capacity characterization for movable antenna
  systems,'' \emph{IEEE Trans. Wirel. Commun.}, vol.~23, no.~4, pp. 3392--3407,
  2024.

\bibitem{wong2020fluid}
K.-K. Wong \emph{et~al.}, ``Fluid antenna systems,'' \emph{IEEE Trans. Wireless
  Commun.}, vol.~20, no.~3, pp. 1950--1962, 2020.

\bibitem{zhu2023movable}
L.~Zhu \emph{et~al.}, ``Movable-antenna enhanced multiuser communication via
  antenna position optimization,'' \emph{IEEE Trans. Wireless Commun.},
  vol.~23, no.~7, pp. 7214--7229, 2024.

\bibitem{zhu2023movable_chanllenges}
------, ``Movable antennas for wireless communication: Opportunities and
  challenges,'' \emph{IEEE Commun. Mag.}, vol.~62, no.~6, pp. 114--120, 2024.

\bibitem{wu2023movable}
Y.~Wu \emph{et~al.}, ``Movable antenna-enhanced multiuser communication:
  Jointly optimal discrete antenna positioning and beamforming,'' in
  \emph{Proc. IEEE Global Commun. Conf.}\hskip 1em plus 0.5em minus 0.4em\relax
  Kuala Lumpur, Malaysia, Dec, 2023, pp. 1--6.

\bibitem{wong2021fluid}
K.-K. Wong and K.-F. Tong, ``Fluid antenna multiple access,'' \emph{IEEE Trans.
  Wireless Commun.}, vol.~21, no.~7, pp. 4801--4815, 2021.

\bibitem{zhuravlev2015experimental}
A.~Zhuravlev \emph{et~al.}, ``Experimental simulation of multi-static radar
  with a pair of separated movable antennas,'' in \emph{Proc. of IEEE Int.
  Conf. Microwaves, Commun. Antennas and Electron. Syst.}, Nov, 2015, pp. 1--5.

\bibitem{basbug2017design}
S.~Basbug, ``Design and synthesis of antenna array with movable elements along
  semicircular paths,'' \emph{IEEE Antennas Wirel. Propag. Lett.}, vol.~16, pp.
  3059--3062, 2017.

\bibitem{mei2024movable}
W.~Mei, X.~Wei, B.~Ning, Z.~Chen, and R.~Zhang, ``Movable-antenna position
  optimization: A graph-based approach,'' \emph{IEEE Wireless Commun. Lett.},
  2024.

\bibitem{wei2024joint}
X.~Wei, W.~Mei, D.~Wang, B.~Ning, and Z.~Chen, ``Joint beamforming and antenna
  position optimization for movable antenna-assisted spectrum sharing,''
  \emph{IEEE Wireless Commun. Lett.}, vol.~13, no.~9, pp. 2502--2506, 2024.

\bibitem{shao20256d}
X.~Shao, R.~Zhang, Q.~Jiang, and R.~Schober, ``{6D} movable antenna enhanced
  wireless network via discrete position and rotation optimization,''
  \emph{IEEE J. Sel. Areas Commun.}, vol.~43, no.~3, pp. 674--687, 2025.

\bibitem{xiao2024channel}
Z.~Xiao \emph{et~al.}, ``Channel estimation for movable antenna communication
  systems: A framework based on compressed sensing,'' \emph{IEEE Trans.
  Wireless Commun.}, 2024, early access, doi={10.1109/TWC.2024.3385110}.

\bibitem{ma2023compressed}
W.~Ma \emph{et~al.}, ``Compressed sensing based channel estimation for movable
  antenna communications,'' \emph{IEEE Commun. Letters}, vol.~27, no.~10, pp.
  2747--2751, 2023.

\bibitem{chen2023joint}
X.~Chen \emph{et~al.}, ``Joint beamforming and antenna movement design for
  moveable antenna systems based on statistical {CSI},'' in \emph{Proc. IEEE
  Global Commun. Conf.}\hskip 1em plus 0.5em minus 0.4em\relax Kuala Lumpur,
  Malaysia, Dec. 2023, pp. 4387--4392.

\bibitem{ye2023fluid}
Y.~Ye \emph{et~al.}, ``Fluid antenna-assisted {MIMO} transmission exploiting
  statistical {CSI},'' \emph{IEEE Commun. Letters}, 2023.

\bibitem{hu2024two}
G.~Hu, Q.~Wu, G.~Li, D.~Xu, K.~Xu, J.~Si, Y.~Cai, and N.~Al-Dhahir,
  ``Two-timescale design for movable antenna array-enabled multiuser uplink
  communications,'' \emph{IEEE Trans. Veh. Technol.}, 2024.

\bibitem{hu2024movable}
G.~Hu, Q.~Wu, D.~Xu, K.~Xu, J.~Si, Y.~Cai, and N.~Al-Dhahir, ``Movable
  antennas-assisted secure transmission without eavesdroppers' instantaneous
  csi,'' \emph{arXiv preprint arXiv:2404.03395}, 2024.

\bibitem{wiesel2007optimization}
A.~Wiesel, Y.~C. Eldar, and S.~S. Shitz, ``Optimization of the mimo compound
  capacity,'' \emph{IEEE Trans. Wireless Commun.}, vol.~6, no.~3, pp.
  1094--1101, 2007.

\bibitem{zheng2009robust}
G.~Zheng, K.-K. Wong, and B.~Ottersten, ``Robust cognitive beamforming with
  bounded channel uncertainties,'' \emph{IEEE Trans. Signal Process.}, vol.~57,
  no.~12, pp. 4871--4881, 2009.

\bibitem{shenouda2007convex}
M.~B. Shenouda and T.~N. Davidson, ``Convex conic formulations of robust
  downlink precoder designs with quality of service constraints,'' \emph{IEEE
  J. Sel. Topics Signal Process.}, vol.~1, no.~4, pp. 714--724, 2007.

\bibitem{xu2022optimal}
D.~Xu \emph{et~al.}, ``Optimal resource allocation design for large
  {IRS}-assisted {SWIPT} systems: A scalable optimization framework,''
  \emph{IEEE Trans. Commun.}, vol.~70, no.~2, pp. 1423--1441, 2022.

\bibitem{sanayei2004antenna}
S.~Sanayei and A.~Nosratinia, ``Antenna selection in {MIMO} systems,''
  \emph{IEEE Commun. Mag.}, vol.~42, no.~10, pp. 68--73, Oct. 2004.

\bibitem{mehanna2013joint}
O.~Mehanna, N.~D. Sidiropoulos, and G.~B. Giannakis, ``Joint multicast
  beamforming and antenna selection,'' \emph{IEEE Trans. Signal Process.},
  vol.~61, no.~10, pp. 2660--2674, 2013.

\bibitem{6698281}
U.~Rashid \emph{et~al.}, ``Joint optimization of source precoding and relay
  beamforming in wireless {MIMO} relay networks,'' \emph{IEEE Trans. Commun.},
  vol.~62, no.~2, pp. 488--499, Feb. 2014.

\bibitem{luo2006introduction}
Z.-Q. Luo and W.~Yu, ``An introduction to convex optimization for
  communications and signal processing,'' \emph{IEEE J. Sel. Areas Commun.},
  vol.~24, no.~8, pp. 1426--1438, 2006.

\bibitem{grant2008cvx}
M.~{Grant} and S.~{Boyd}, ``{CVX}: Matlab software for disciplined convex
  programming, version 2.1,'' \emph{http://cvxr.com/cvx}, Jan. 2020.

\bibitem{horst2013global}
R.~Horst and H.~Tuy, \emph{Global optimization: Deterministic
  approaches}.\hskip 1em plus 0.5em minus 0.4em\relax Springer Science \&
  Business Media, 2013.

\bibitem{boyd2007branch}
S.~Boyd and J.~Mattingley, ``Branch and bound methods,'' \emph{Notes for
  EE364b, Stanford University}, vol. 2006, p.~07, 2007.

\bibitem{le2012exact}
L.~Thi \emph{et~al.}, ``Exact penalty and error bounds in {DC} programming,''
  \emph{Journal of Global Optimization}, vol.~52, no.~3, pp. 509--535, 2012.

\bibitem{ng2015secure}
D.~W.~K. Ng and R.~Schober, ``Secure and green {SWIPT} in distributed antenna
  networks with limited backhaul capacity,'' \emph{IEEE Trans. Wireless
  Commun.}, vol.~14, no.~9, pp. 5082--5097, 2015.

\bibitem{bomze2010interior}
I.~M. Bomze \emph{et~al.}, ``Interior point methods for nonlinear
  optimization,'' \emph{Nonlinear Optimization: Lectures given at the CIME
  Summer School held in Cetraro, Italy, July 1-7, 2007}, pp. 215--276,
  2010\color{black}.

\bibitem{boyd2004convex}
S.~P. Boyd and L.~Vandenberghe, \emph{Convex optimization}.\hskip 1em plus
  0.5em minus 0.4em\relax Cambridge University Press, 2004.

\bibitem{razaviyayn2013unified}
M.~Razaviyayn \emph{et~al.}, ``A unified convergence analysis of block
  successive minimization methods for nonsmooth optimization,'' \emph{SIAM
  Journal on Optimization}, vol.~23, no.~2, pp. 1126--1153, 2013\color{black}.

\bibitem{yu2021robust}
X.~Yu \emph{et~al.}, ``{IRS}-assisted green communication systems: Provable
  convergence and robust optimization,'' \emph{IEEE Trans. Commun.}, vol.~69,
  no.~9, pp. 6313--6329, Sept. 2021.

\bibitem{wu2024globally}
Y.~Wu \emph{et~al.}, ``Globally optimal movable antenna-enhanced multi-user
  communication: Discrete antenna positioning, motion power consumption, and
  imperfect csi,'' \emph{arXiv preprint arXiv:2408.15435}, 2024.

\bibitem{AM3248}
\BIBentryALTinterwordspacing
{Faulhaber}, \emph{{Motors, Series AM2224}}, 2023. [Online]. Available:
  \url{{https://www.faulhaber.com/fileadmin/Import/Media/EN_AM2224_FPS.pdf}}
\BIBentrySTDinterwordspacing

\bibitem{ning2024movable}
B.~Ning \emph{et~al.}, ``Movable antenna-enhanced wireless communications:
  General architectures and implementation methods,'' \emph{IEEE Wireless
  Communications}, pp. 1--9, 2025.

\bibitem{mclean1988review}
G.~McLean, ``Review of recent progress in linear motors,'' in \emph{IEE
  Proceedings B (Electric Power Applications)}, vol. 135, no.~6.\hskip 1em plus
  0.5em minus 0.4em\relax IET, 1988, pp. 380--416.

\bibitem{yu2016alternating}
X.~Yu, J.-C. Shen, J.~Zhang, and K.~B. Letaief, ``Alternating minimization
  algorithms for hybrid precoding in millimeter wave mimo systems,'' \emph{IEEE
  J. Sel. Topics Signal Process.}, vol.~10, no.~3, pp. 485--500, 2016.

\bibitem{boumal2014manopt}
N.~Boumal, B.~Mishra, P.-A. Absil, and R.~Sepulchre, ``Manopt, a matlab toolbox
  for optimization on manifolds,'' \emph{The Journal of Machine Learning
  Research}, vol.~15, no.~1, pp. 1455--1459, 2014.

\bibitem{chen2018Mutualcoupling}
X.~Chen, S.~Zhang, and Q.~Li, ``A review of mutual coupling in mimo systems,''
  \emph{Ieee Access}, vol.~6, pp. 24\,706--24\,719, 2018.

\bibitem{savy2016couplingeffect}
L.~Savy and M.~Lesturgie, ``Coupling effects in mimo phased array,'' in
  \emph{2016 IEEE Radar Conference (RadarConf)}.\hskip 1em plus 0.5em minus
  0.4em\relax IEEE, 2016, pp. 1--6.

\bibitem{liao2012adaptive}
B.~Liao and S.-C. Chan, ``Adaptive beamforming for uniform linear arrays with
  unknown mutual coupling,'' \emph{IEEE Trans. Antennas Propag.}, vol.~11, pp.
  464--467, 2012.

\end{thebibliography}

\end{document}